\documentclass[12pt]{article}

\usepackage{url}
\usepackage[final]{graphicx}
\usepackage{pgf}
\usepackage{tikz}
\usepackage{pgfplots}
\pgfplotsset{width=7cm}
\usetikzlibrary{arrows,automata}
\usepackage{amssymb,amsmath}
\usepackage{ctable}
\usepackage{hyperref}
\usepackage{multicol}
\usepackage{wrapfig}
\usepackage{algorithm}
\usepackage{algorithmic}
\usepackage{appendix}
\usepackage{xfrac}
\usepgfplotslibrary{external} 
\newcommand{\ind}[1]{1\hspace{-.55ex}\mbox{l}_{\{#1\}}}  %fonction indicatrice
\newcommand{\define}{:=}
\newcommand{\vect}[1]{\mathbf{#1}}

\newtheorem{theorem}{Theorem}[section]
\newtheorem{lemma}[theorem]{Lemma}
\newtheorem{proposition}[theorem]{Proposition}

\newtheorem{assumption}{Assumption}

\newenvironment{proof}[1][Proof]{\begin{trivlist}
\item[\hskip \labelsep {\bfseries #1}]}{\end{trivlist}}

\makeatother

\author{Shreyas Sekar\\ IIT Roorkee, India \and Patrick Maill\'e\\ Telecom Bretagne, France}

\title{\vspace{-2cm}Exploiting the `Web of Trust' to improve efficiency in collaborative networks}

\begin{document}
\date{\today}

\maketitle
\begin{abstract}
Maintaining high quality content is one of the foremost objectives of any web-based collaborative service that depends on a large number of users. In such systems, it is nearly impossible for automated scripts to judge semantics as it is to expect all editors to review the content. This catalyzes the need for trust-based mechanisms to ensure quality of an article immediately after an edit. In this paper, we build on previous work %~\cite{maille2011trust} 
and develop a framework based on the `web of trust' concept to calculate satisfaction scores for all users without the need for perusing the article. We derive some bounds for systems based on our mechanism and show that the optimization problem of selecting the best users to review an article is NP-Hard. Extensive simulations validate our model and results, and show that trust-based mechanisms are essential to improve efficiency in any online collaborative editing platform.

\end{abstract}
\section{Introduction}

The emergence and evolution of the world wide web has shifted focus towards services that follow a dynamic and interactive paradigm~\cite{craig2007changing, hoegg2006overview}. Perhaps, the most prominent among these are collaborative platforms that have revolutionized the physics of content generation and communication. These include, but are not limited to wikis, weblogs, versioning software and real-time document editing suites. Services like Wikia~\cite{wikia} allow just about anyone to create a new wiki on any desired theme and indeed, wikis exist on topics from Vintage Sewing~\cite{vintagesew} to Star Wars~\cite{starwars}.

Unfortunately, the open-to-edit nature of these systems has led to serious quality breaches in the past~\cite{vandalismwiki}, deterring their use in academic and professional contexts~\cite{snyder2007sa}. Ensuring that 
the content conforms to minimum quality standards is essential to enable the widespread adoption of these platforms. However taking into consideration the dynamic nature and large user base of such communities, it is not practical to expect every user to review or edit the document before labeling it as `reliable'. At the same time, different users have different satisfiabilities and expectations and it is not a trivial task to translate the satisfaction level of one user to the rest of the community. Thus, an important challenge is to quantify content quality and accuracy in terms of the satisfaction levels of different collaborators and determine when the document is ready for publication while using minimal human effort. One metric that expedites this process is trust.

Recent work has focused on exploiting the underlying social network structure in online communities. This has led to the development of trust-based mechanisms for recommendation systems~\cite{andersen2008trust, rozenfeld2009consistent, walter2008model}, peer-to-peer networks~\cite{xiong2004peertrust}, Internet transactions~\cite{josang2007survey} and other general web-based networks~\cite{golbeck2005computing}. A common approach is to use a trust propagation model similar to the one propounded in~\cite{guha2004propagation} to discern user preferences, i.e. for a given item, based on the feedback of users who have already tried it, determine whether or not to recommend it to an arbitrary user.

The problem of estimating satisfaction in networks is also similar to that of information propagation/diffusion, which has been the subject of a rich literature~\cite{gruhl2004information}. Kempe et al~\cite{kempe2003maximizing} considered the problem of selecting the most influential nodes in a graph to market an item to, in order to ensure its maximal spread and showed that the optimization version of this problem is NP-Hard. These works, however do not answer a few critical questions pertinent to collaborative document-editing: how to determine when the document is ready for publishing and if not, how to choose the next person to edit or estimate its quality. 	

In~\cite{maille2011trust}, the authors introduce the notion of trust for collaborative work. They define a measure known as \textit{satisfaction score}, which is the estimated satisfaction or rating of a user at the current stage of document development, when the user has not yet read the article. Their model assumes that every user $i$ has a trust value for every other user $j$ ($t_{ij} \in [0,1]$), apart from a unique threshold ($b_i \in [0,1]$). If the satisfaction score of a user is below his %\footnote{Throughout the article his should be taken to mean his/her}
 threshold, then he is considered to be `unsatisfied'. An editor $j$ upon reviewing the article gives it an unbiased rating $e_j$ and the satisfaction score of an arbitrary user $i$ is calculated as follows: 
\begin{equation}
s_i = e_j \times t_{ij}
\label{sat1d}
\end{equation}
 Various mechanisms were proposed to select the successive editor. Document editing ends when all users are satisfied with the article (i.e. $s_i > b_i,  \forall i$).

In this paper, we propose a generalized mechanism to calculate satisfaction scores for users in arbitrary social networks, and use this to select the successive reviewers for an article. We use the term `web of trust' to denote social networks where users trust only those nodes with whom they have interacted previously, and this trust is quantifiable. The situation of interest here would be one where a new article has to be reviewed so that sufficient quality approvals (in terms of satisfaction among all users) can be reached. As regards the computation of satisfaction scores, the model proposed in~\cite{maille2011trust} emerges as a special case of our model when the underlying graph is a complete graph. We later show that the social network structure can be explicit (in terms of contacts and followers) or implicitly derived based on similarity of opinion and content. 

In order to analyze the efficiency gained by using our model, we resort to random graphs where each edge has a non-zero probability $p$ of existing. Random graphs, first proposed by Erd\"{o}s and R{\'e}nyi~\cite{erdos1960evolution}, have found extensive applications in the past especially in the case of epidemic propagation~\cite{bollobas2001random,durrett2010some}. We derive some properties of our model in random graphs and use them to show a bound on the number of users required to review an article, in the case when it is of acceptable quality. Although graph models based on the Small World~\cite{watts1998collective} and Power law~\cite{barabási2000scale} distributions are more realistic representations of social networks, Erd\"{o}s-Renyi graphs allow us to gain a better understanding of the parameters governing the satisfaction distribution while allowing some analytical reasoning. %without much focus on the effects of the network topology. 
Our main objective here is to use trust to minimize the human evaluation effort required to be expended at each stage of the document development. It is pertinent to remark here that we are concerned only with the trust between pairs of users rather than global trust better known as reputation~\cite{resnick2000reputation}. This has been done giving due consideration to the subjective nature of content on the Internet as not all collaborative platforms have factual accuracy as their sole aim. For instance, interpretations of most works of fiction are highly subjective.

The remainder of the paper is organized as follows. In Sections~\ref{sec:model} and~\ref{sec:satisf_model}, we define the idea of trust for a collaborative platform and propose a model based on trust to calculate the satisfaction score of each user with respect to a single article. We also prove uniqueness and existence of satisfaction scores and present an  efficient algorithm to calculate the scores using trust matrices. In Section~\ref{sec:random_graphs} we consider the random graph model, and prove bounds on the minimum number of users required to review the article based on the expected satisfaction score. We prove that the optimization problem of selecting the best users to review the article is NP-Hard in Section~\ref{sec:select_raters}, and present an improvised greedy algorithm to choose the best potential raters. Finally, in Section~\ref{sec:simulations} we validate our results and observations via simulation and discuss future work in Section~\ref{sec:conclusion}.

\subsection{Additional Related work}
The concept of trust in computation has been inspired by every day human relationships and has found applications in numerous fields including medical information~\cite{blaze1996managing}, mobile~\cite{wilhelm1998problem} and pervasive computing~\cite{sun2008performance}, and security~\cite{blaze1999role}. However, the growth of the world wide web has resulted in an alarmingly large number of transactions between complete strangers and it has become imperative to utilise feedback and transactional history to develop trust or reputation based models, the most prominent being the ones on Epinions~\cite{massa2005controversial} and eBay~\cite{resnick2002trust}. Nowadays, personalized content is almost ubiquitous on the web in the form of advertisements, search results, movie recommendations (See Netflix \$1 million challenge~\cite{lohr2009netflix}), interesting links, etc. Many of these systems consider only binary product opinion (like/dislike) and thus the results obtained are not valid when ratings vary across a continuous scale. Recommendation systems are also not conservative, in the sense that their motive is to exactly predict a non-rater's vote preference for a given item as opposed to being cautious in the case of satisfaction scores.

Our work is most similar to that of Rozenfeld and Tennenholtz~\cite{rozenfeld2009consistent}, who propose a continuous,  consistent recommendation system. The model proposed in Section~\ref{sec:satisf_model} can be viewed as a generalization of their recommendation system. However a few of their properties are not very pertinent to the case of collaborative editing, though we show that our specific function satisfies most properties like consistency and monotonicity. 

Network propagation models have received extensive consideration in the contexts of epidemiology~\cite{ganesh2005effect}, (mis)information~\cite{acemoglu2010spread}, cascading faults~\cite{asavathiratham2001influence}, patches~\cite{vojnovic2008race}, etc. A majority of these models overlook the need for trust and concentrate more on the temporal aspects of information propagation than on information intensity and attenuation of information reliability as one moves away from the source. The generalized threshold model in~\cite{kempe2003maximizing} considers a threshold function $f_v$ for every user ($v$) which maps all subsets of $v$'s neighbor set to arbitrary values in $[0,1]$. Our model can be viewed as an alternative realization of the same with fixed thresholds and continuous satisfaction scores. The problem of selecting the most influential nodes to market a product to is of great significance to economics and data mining and was first posed formally by Domingos and Richardson~\cite{domingos2001mining} sparking off a flurry of research. Subsequent work~\cite{chang2009spreading, kempe2003maximizing, lappas2010finding} has proven that most varieties of the optimization problem of selecting the best seeders is NP-Hard. An excellent survey of information propagation models and the ``Maximizing influence" problem can be found in~\cite{wortman2008viral}.

With respect to quality of collaborative content, the recently proposed WikiTrust~\cite{adler2008assigning} uses author reputation and number of edits to measure the trustworthiness of each word in a wiki article, and detect vandalism. Automated evaluations based on global reputation and article semantics are however beyond the scope of this paper and we stick to using reviews from a small subset of the user base to calculate satisfaction scores for the rest. As far as we are aware, this is the first study applying the web of trust model for collaborative work. 

\section{Model}\label{sec:model}
\subsection{Trust}
\label{trustdefinition}
Although the notion of trust in computation has been inspired heavily by real-life relationships, its definition is application-dependent. For instance, it is pertinent to define trust in recommendation systems based on `similarity of tastes and preferences' as it makes sense in P2P applications to define it based on a user's bandwidth and upload/download ratio. For the purpose of collaborative editing, we derive inspiration from the definition proposed by Sztompka~\cite{sztompka1999trust}:
\begin{quote}
``Trust is a bet about the future contingent actions of others",
\end{quote}
which is interpreted in our application as: \textit{Trust is the amount of faith a user has in the choices and actions of another user}. This definition is consistent with Equation~(\ref{sat1d}) for calculating satisfaction scores based on the rating of a single user (the same user who had edited the article). The model presented in this paper is based on classically accepted assumptions and notations:
\begin{enumerate}
\item $t_{ij}$ represents the trust user $i$ has in the actions of user $j$. This can be viewed as a directed edge from $i$ to $j$ which represents the direction of trust ($i$ trusts $j$ and hence draws information about the article from $j$'s assessment to make an estimation about his own satisfaction).
\item $\forall i,j$, $0 \leq t_{ij} \leq 1$, wherein a trust of $0$ indicates that either $i$ has no knowledge of $j$ or does not trust $j$ at all. The absence of an edge from $i$ to $j$ indicates that $i$ does not know (and hence cannot trust) $j$, i.e. he will ignore $j$'s assessment in his satisfaction estimation. A trust score of $1$ indicates that $i$ has complete faith in $j$'s actions.
%\item By definition, $t_{ii}=1$.
\item $t_{ij}$ is not necessarily equal to $t_{ji}$. This models the asymmetry that exists in real-life relationships.
\end{enumerate}

We relax two key assumptions made in~\cite{maille2011trust} in our model
\begin{enumerate}
\item Having a trust value between every pair of users imposes the restriction that every collaborator must have some knowledge about the reliability of every other collaborator. This is however, impossible in large online communities where the total number of user pairs is of the order $O(N^2)$, where $N$  is the number of collaborators. In our model, users have trust values only for a subset of the total user base, i.e. other users about whom they have prior knowledge.
\item Instead of the article being rated by only the previous editor, our model allows a small subset of the collaborators to read and rate the article, and if necessary edit it in which case the article has to be rated again by other interested users. 
\end{enumerate}
Thus, one can view the system as a directed graph where an edge between a pair of users indicates existing trust and thereby some sort of prior interaction. Figure~\ref{samplegraph} shows the sample graph of relationships in a network with six users. The shaded circles represent the raters of the current version of the article. %\shre{For the sake of simplicity, self-loops are not shown in the figure}\remshre{I propose removal of this line}. 

\begin{figure}[htbp]
\centering
\begin{tikzpicture}[shorten >=1pt,->]
  \tikzstyle{vertex}=[circle,fill=blue!25,minimum size=16pt,inner sep=0pt]

   \foreach \name/\angle/\text/\xshift in {A/162/A/4cm, D/162/D/0cm, 2
                                  E/230/E/1.5cm, F/162/F/5.8cm}
    \node[vertex,xshift=\xshift] (\name) at (\angle:1cm) {$\text$};

\tikzstyle{vertex}=[circle,fill=black!100,minimum size=16pt,inner sep=0pt,text=white!100]
  \foreach \name/\angle/\text/\xshift in {C/90/C/1.5cm, B/240/B/4.5cm}
    \node[vertex,xshift=\xshift,yshift=0cm] (\name) at (\angle:1cm) {$\text$};      

 \path (A) edge []  node {} (B)
		   edge  [bend right] node {} (F)
       (C) edge  node {} (A)
			edge   node {} (D)
			(F) edge node {} (A)
			(D) edge [bend left] node {} (C)
			edge   node {} (A);

            	 \end{tikzpicture}
            	 \caption{Sample Graph representing a community with 6 users and 2 raters}
            	 \label{samplegraph}
\end{figure}
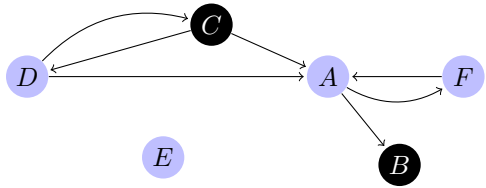

The above assumptions directly yield our problem statement:
\begin{quote}
\textbf{Problem}\\
Given a large community of users, all interested in the outcome of a single document or entity, and a set of trust values for each user on other users, use the quality evaluation of the document provided by a few users (raters) to estimate the satisfaction of the others (non-raters).
\end{quote}

%\subsection*{Notation}
%We introduce some notation before presenting our satisfaction estimation model.\\

In the rest of the paper, we will denote the social (trust) network by the triplet $G\define((V,E),T,B)$, where $V,E$ are the vertices and edges respectively of a directed graph with $|V|=N$, each node being a collaborator of the document and edges representing trust between collaborators\footnote{We use the terms user and collaborator interchangeably throughout this paper.}. For the sake of consistency, we shall use $(i,j) \in E$ to represent a directed edge from $i$ to $j$, meaning that $i$ trusts $j$ (to some extent). $T:E \rightarrow [0,1]$, referred to as the trust matrix, consists of the weights for each edge, i.e. trust between users. As mentioned previously, we shall use $t_{ij}$ to denote the trust user $i$ has on user $j$.  Finally, $B: V \rightarrow [0,1]$ is a set of threshold values that denote the minimum quality expectation of each user regarding the article.

We shall also use the notation $G_R=((V,E),T,B,N_R,R)$ to include the % is a five-tuple representing the state of the underlying social network and 
details of the raters in the situation description. By rater, we refer to a collaborator who has read the document and given it a quality estimation or rating on a scale of $1$.  $N_R \subset V$ represents the set of raters.
%We use $k=\frac{|N_R|}{N}$ to denote the fraction of raters with respect to the total number of users.\rempat{this is not needed right now I think}
$R: N_R \rightarrow [0,1]$ represents the rating associated with each rater. Abusing notation, we shall henceforth use $r_i=R(i)$ to denote the rating of user $i$, if he is a rater and $b_i=B(i)$ to denote his quality threshold. 
%Similarly if we use ${\vect s}$ to denote a column vector of all satisfaction values, then $s_j$ refers to the satisfaction of user $j$ that is to be estimated. 

Let $N(i)$ represents all nodes that user $i$ trusts, i.e. $N(i)=\{ j | (i,j)\ \in E\}$. Then we denote by $N_R(i)$, the neighbors of $i$ who are raters and $N_N(i)$, the non-raters, i.e. $N_R(i) = N(i) \cap N_R$ and $N_N(i) = N(i)\setminus N_R$. Finally, a user $i$ is said to be \emph{satisfied} if his satisfaction level is larger than his quality threshold, i.e. $s_i > b_i$.

\section{Satisfaction Estimation Model}\label{sec:satisf_model}
\label{model}
Intuitively, the satisfaction level of a non-rater should depend not only on the ratings of the various raters that he is connected to, but also on the transitivity of trust values as ratings spread across the network. The main question that arises is: if a non-rater does not directly trust a certain rater, then how can he take into account that rating in his own satisfaction calculation? This question has received widespread attention in the literature~\cite{guha2004propagation,richardson2003trust}, as propagation of (mis)trust is a key issue in most trust-based networks. 
For instance, let us consider the graph of Figure~\ref{samplegraph} induced on only the three users $[A,B,D]$.

\begin{wrapfigure}{r}{0.5\textwidth}
      \centering
\begin{tikzpicture}[shorten >=1pt,->]
  \tikzstyle{vertex}=[circle,fill=blue!25,minimum size=16pt,inner sep=0pt]

   \foreach \name/\angle/\text/\xshift in {A/162/A/2cm, D/162/D/0cm}
    \node[vertex,xshift=\xshift] (\name) at (\angle:1cm) {$\text$};

\tikzstyle{vertex}=[circle,fill=black!100,minimum size=16pt,inner sep=0pt,text=white!100]
  \foreach \name/\angle/\text/\xshift in {B/240/B/2.5cm}
    \node[vertex,xshift=\xshift,yshift=0cm] (\name) at (\angle:1cm) {$\text$};      

 \path (A) edge  node {} (B)
       (D) edge  node {} (A);
            	 \end{tikzpicture}
\label{inducedgraph}
\caption{The network of Figure~\ref{samplegraph} induced on the three users A, B, D}
\end{wrapfigure}
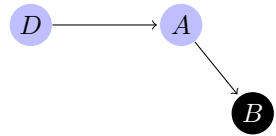
From the figure, and following~(\ref{sat1d}), A's satisfaction should be proportional to $r_B t_{AB}$. On the other hand, since D does not directly know B, he can only estimate his satisfaction based on A's satisfaction, thus depending on A's belief about B. This can be interpreted as $s_D = (t_{DA} \odot t_{AB})r_B$, where $\odot$ may be any appropriate binary operator depending on the application. In this work, we take $\odot = \times$, the binary multiplication operator in order to ensure that our system is conservative (to be defined later). Thus, $s_D = (t_{DA} t_{AB})r_B$.
Indeed, in this paper trust is used to factor a certain level of caution when basing one's decision on the recommendations of others. Overestimating one's satisfaction score because of false recommendations must in particular be avoided; hence the multiplication by the trust value in the above satisfaction calculation. 

However, this problem becomes much more complicated when there are multiple raters and multiple paths between a rater and non-rater. It is common practice to tackle this problem by estimating (indirect) trust between all pairs of users. However, in this work we avoid indirect trust estimation and instead, compute the satisfaction score of a non-rater as a weighted average of the (discounted by the trusts) scores of its trusted peers, using the following generalized model.

%\begin{equation}
%s_i = \frac{ \displaystyle\sum\limits_{j_1 \in N_R(i)} f_1(s_{j_1},t_{ij_1}) + \displaystyle\sum\limits_{j_2 \in N_N(i)} f_2(s_{j_2},t_{ij_2}) }{\displaystyle\sum\limits_{j_1 \in N_R(i)} g_1(t_{ij_1}) + \displaystyle\sum\limits_{j_2 \in N_N(i)} g_2(t_{ij_2})}
%\end{equation}
%where $f_i$ and $g_i$ ($i=1,2$) are generic functions.
\begin{equation}
s_i = \frac{ \displaystyle\sum\limits_{j_1 \in N_R(i)} f(t_{ij_1})t_{ij_1}s_{j_1} + \displaystyle\sum\limits_{j_2 \in N_N(i)} g(t_{ij_2})t_{ij_2} s_{j_2}}{\displaystyle\sum_{j_1 \in N_R(i)} f(t_{ij_1}) + \displaystyle\sum\limits_{j_2 \in N_N(i)} g(t_{ij_2})}
\end{equation}
%where $f_i$ and $g_i$ ($i=1,2$) are generic functions.
where the weight functions $f$ and $g$ are generic.

Based on various requirements (to be elucidated in the following section), we propose and analyze the following specific equation in this paper:
\begin{equation}
   s_i = \left\{ 
   \begin{array}{c l}
     \frac{ \displaystyle\alpha\sum\limits_{j_1 \in N_R(i)} t_{ij_1}^2s_{j_1} + \displaystyle(1-\alpha)\sum\limits_{j_2 \in N_N(i)} t_{ij_2}^2 s_{j_2}}{\displaystyle \alpha \sum\limits_{j_1 \in N_R(i)} t_{ij_1} + (1-\alpha) \displaystyle\sum\limits_{j_2 \in N_N(i)} t_{ij_2}} & \quad \textrm{if $i$ is a non-rater}\\ \\
     r_i & \quad \textrm{if $i$ is a rater,}\\
   \end{array} \right.
   \label{sat2d}
\end{equation}
where $\alpha$ is a parameter which determines the relative importance of raters with respect to non-raters. In this model, we take $\sfrac{1}{2} \leq \alpha \leq 1$, in order to provide greater weight to the opinions of a person who has directly read the article than someone whose opinions are based on hearsay.

To simplify the expression, we write the weights used in~(\ref{sat2d}) as follows:
\begin{equation}\label{eq:weights}
\forall i\in N,\forall j\in N(i),\qquad w_{ij}\define t_{ij}^2\frac{\alpha \ind{j\in N_R}+(1-\alpha)\ind{j \notin N_R}}{\displaystyle \alpha \sum\limits_{j_1 \in N_R(i)} t_{ij_1} + (1-\alpha) \displaystyle\sum\limits_{j_2 \in N_N(i)} t_{ij_2}},
\end{equation}
where $\ind{X}=1$ if $X$ is satisfied and $0$ otherwise. With this notation,~(\ref{sat2d}) implies in particular that
\begin{equation}\label{eq:nonrater}
i\notin N_R\quad\Rightarrow\quad s_i=\sum\limits_{j\in N(i)}w_{ij}s_j.
\end{equation}

Mathematically, the model can be viewed as follows. Suppose a node $A$ has neighbors with satisfaction scores of $\{s_1, s_2, s_3, s_4\}$ (each neighbor can be a rater or a non-rater). %with its trust on the same neighbors being $\{t_{1}, t_2, t_3, t_4\}$ respectively. 
Then, in accordance with Equation~(\ref{sat1d}) the node receives values $\{t_{A1}s_1, t_{A2}s_2, t_{A3}s_3, t_{A4}s_4\}$ as possible satisfaction scores. Thus it is reasonable to take the actual satisfaction level of the node as a weighted mean of the incoming values from each of its neighbors. Here, we take the weight of the $j^{th}$ satisfaction recommendation ($t_{Aj}s_j$) to be the trust value on the corresponding neighbor ($t_{Aj}$). Weighted averaging with priority $\alpha$ for raters and $(1-\alpha)$ for non-raters yields Equation~(\ref{sat2d}) for satisfaction calculation. It is also apparent that when the underlying social network is a complete graph (i.e. $E=(V\times V)$, which occurs when all collaborators have prior knowledge about each other), there is only one rater ($|N_R| = 1$) and $\alpha=1$ (zero weightage to satisfaction scores of non-raters), Equation~(\ref{sat2d}) simplifies to Equation~(\ref{sat1d}). Thus, the satisfaction estimation model proposed here is an extended version of the model in~\cite{maille2011trust}.
%Write about how the equation 1 arrives as a special case of this model when alpha=1

The above formula yields a system of linear equations with $N-|N_R|$ unknowns, where the unknown variables are the satisfaction scores ($s_j$). A necessary condition for a non-trivial solution to exist is that $|N_R| > 0$, i.e. there must exist at least one collaborator who has read and rated the article. Drawing a parallel to network diffusion algorithms, one may attempt to calculate the spread of satisfaction scores across the network\footnote{Starting from the source(s), the ratings spread across the network, and the satisfaction of a node at a distance $k$ from the source is calculated at the $k^{th}$ iteration. This technique is analogous to breadth-first search.}. It is not very clear as to whether such an algorithm would converge as the satisfaction score of a user is dependent on the satisfaction scores of his neighbors, which may or may not be unknowns themselves at any given stage. The worst case is clear from Figure~\ref{samplegraph}, where the satisfaction score of $A$ depends on $F$, which in return depends on $A$'s score itself. Moreover, there is also some ambiguity over what the satisfaction of a non-rater should be when there is no path between him and a rater. These difficulties are tackled in the following section where a simple, yet computationally efficient iterative algorithm is presented to calculate satisfaction scores.

We are now in a position to describe the complete collaborative editing process by means of a mechanism that uses the satisfaction estimation model defined above. The mechanism is managed by a central authority which selects the most appropriate user to read/edit the document at each stage and also maintains the trust values between users. It is assumed that all users are interested only in improving the quality of the article and are without malicious intent, though we shall later present a trust updation mechanism that provides some incentive for non-vandalistic behaviour. The following pseudocode details the collaborative editing process:

\begin{algorithm}[htbp]
\caption{Collaborative Editing Scheme}
\label{alg1}
\algsetup{indent=2em}
\begin{algorithmic}[1]
\STATE $D \leftarrow$ Document 
\STATE $D(0) \leftarrow$ Initial State of Document
\STATE $n \leftarrow 0$
\STATE $s_j \leftarrow 0 \quad \forall j \quad \quad \quad$ \COMMENT{\textit{Set all users to be unsatisfied}}
\WHILE{$\exists j, s_j\leq b_j$} \label{alg1:condition}
\STATE \textbf{select} user ($i$) to read and rate the document
\IF{$D$ has been edited}
\STATE $D(n+1) \leftarrow$ Updated Document
\STATE $n \leftarrow n+1$
\STATE $s_j \leftarrow 0 \quad \forall j \quad$ \COMMENT{\textit{Reset satisfaction values to zero}}
\ELSE
\STATE $r_i \leftarrow$ Rating of Article
\STATE Update Satisfaction Scores of all non-raters using Equation~(\ref{sat2d})
\ENDIF
\ENDWHILE
\end{algorithmic}
\end{algorithm}

Simply put, as long as there exists at least one collaborator who is unsatisfied with the document, the algorithm chooses a user to rate (and potentially edit if necessary) the document. %, thereby ensuring convergence of the process. 
The only case where a deadlock might arise in the process is when a rater gives the article a rating lower than his own threshold value. However, such a case is not likely to arise in practical situations as a collaborator who is unsatisfied with the article quality will edit it himself in order to improve upon the content he is dissatisfied with. We remark here that if a constant $0 \leq \eta\leq 1$ is defined to be the fraction of satisfied users required in order to publish the document, then the algorithm can be modified (at line \ref{alg1:condition}) to continue the loop if the fraction of satisfied users is less than $\eta$ or else end the editing process and publish the document.

\subsection*{Properties satisfied by the model}

The model defined satisfies a few  properties which make it desirable for collaborative publishing. The definition of most of these properties  are based on those listed in~\cite{rozenfeld2009consistent}, but are modified to reflect the needs of collaborative document editing. 
\begin{enumerate}
\item \textbf{Stability of Satisfaction:} The satisfaction score of a user is altered iff a non-rater becomes a rater, or the article content changes. This \textit{stability} of scores holds because trust values remain constant throughout the collaborative editing process. Later we propose a trust updation mechanism that maintains this property by updating trust scores only between raters.

\item \textbf{Bounded Satisfaction:} $0 \leq s_i \leq 1 \quad \forall i$, irrespective of $G$. This is true because $0 \leq t_{ij} \leq 1$ and $0 \leq r_j \leq 1 \quad \forall i,j$ and thus any linear combination of them should lie within the same limits. 

\item \textbf{Conservativeness:} For any node, its satisfaction score lies between the maximum and minimum satisfaction of its neighbors (multiplied by the trust in them). Mathematically, $ \bar{s}_{i,min}\leq s_i \leq \bar{s}_{i,max}$, where $\bar{s}_{i,min} = \min_{\forall j}\{{t_{ij}s_j | (i,j) \in E}\}$ and $\bar{s}_{i,max} = \max_{\forall j}\{{t_{ij}s_j | (i,j) \in E}\}$. As a corollary, the satisfaction score of any user is not greater than the maximum rating given by a rater. 
%This is in fact a very weak property as in general $ \min_{\forall j}\{{t_{ij}s_j}\}\leq s_i \leq \max_{\forall j}\{{t_{ij}s_j}\}$. 
Conservativeness can be interpreted as, while a node receives several recommendations for its satisfaction score, it is extremely cautious when deciding upon the final satisfaction value. In order to avoid over-estimation of the document, the satisfaction scores of its neighbors are brought down by a factor equal to its trust in them. Thus the satisfaction score of a user can be viewed as a lower bound on his quality estimation of the document.

\item \textbf{Consistency and Progressiveness:} In the event that a non-rater becomes a rater and provides the article a rating not less than his last estimated satisfaction score, then the satisfaction of all users in the network either increase or remain the same. Formally, if $s_j$ is the satisfaction of user $j$ in the system $G_R=((V,E),T,B,N_R,R)$, and $s'_j$ is the satisfaction under $G_R'=((V,E),T,B,N_R\cup i,R\cup r_i) \mbox{ for any } i, \mbox{ such that } r_i \geq s_i \mbox{ then } s'_j \geq s_j$. \textit{Progressiveness} is the property that as long as the document is not edited, the satisfaction score of a user is strictly nondecreasing, assuming that the score always remains a lower bound for the actual rating of a user. This can be interpreted as when more users review a (good) article, there is a natural tendency to become less conservative and increase one's lower bound.

\item \textbf{Independence:}
\begin{enumerate}
\item \textbf{From Disconnected users:} Removing a node to whom a particular user has no path whatsoever, and all its associated edges does not affect the user's satisfaction.

\item \textbf{Between Rater and Rater:} Removing edges between raters does not alter the satisfaction score of any user in the network, for that particular stage of document editing.

\end{enumerate}

\item \textbf{Irrelevance of Order:} The final satisfaction of a user does not depend on the order in which the raters were chosen. Thus the system depends only on the current input state $G_R=((V,E),T,B,N_R,R)$.
\end{enumerate}

\subsection{Existence and uniqueness of the satisfaction scores}
\label{sec:exist}

We prove here that the system of linear equations expressed by~(\ref{sat2d}) for non-raters admits a unique solution, under some mild assumptions that we now clarify.
The first convention we take, still in the spirit of users being cautious while considering document quality, concerns non-raters who have no \textit{trust path} (i.e. no path consisting of edges in $E$)
connecting them to any rater.
\begin{assumption}\label{as:virgin_score}
If a user does not trust (even indirectly) any rater, then his satisfaction score is $0$.
\end{assumption}

We can then prove the following result.
\begin{proposition}\label{prop:exist_uniq}
For a given trust network $G_R=((V,E),T,B,N_R,R)$, the set of equations~(\ref{sat2d}) defines a unique vector $(s_i)_{i\in V}$ satisfying Assumption~\ref{as:virgin_score}.
\end{proposition}

\begin{proof}

As before, we denote the rating of a user $i\in N_R$ by $r_i=R(i)$.

\noindent{\bf Existence.} 
Consider the series of size-$N$ vectors $(\vect{s}^{(n)})_{n\in{\mathbb N}}$ defined for all $i$ by:
\begin{eqnarray}
s_i^{(0)}&\!\!\!=\!\!\!&\left\{\begin{array}{ll}r_i&\text{ if }i\in N_R\\0&\text{ otherwise,}\end{array}\right.\label{eq:s0}\\
\forall n\geq 0,\quad s_i^{(n+1)}&\!\!\!=\!\!\!& \left\{ 
   \begin{array}{l l}
     r_i & \quad \textrm{if $i\in N_R$,}\\ \\
     \frac{ \displaystyle\alpha\!\!\!\sum\limits_{j_1 \in N_R(i)}\!\!\! t_{ij_1}^2s_{j_1}^{(n)} + \displaystyle(1-\alpha)\!\!\!\sum\limits_{j_2 \in N_N(i)}\!\!\! t_{ij_2}^2 s_{j_2}^{(n)}}{\displaystyle \alpha \!\!\!\sum\limits_{j_1 \in N_R(i)}\!\!\! t_{ij_1} + (1-\alpha)\!\!\!\sum\limits_{j_2 \in N_N(i)}\!\!\! t_{ij_2}} & \quad \textrm{otherwise}
   \end{array} \right.\label{eq:series}
\end{eqnarray}
%Remark that~(\ref{eq:series}) can also be written as $\vect{s}^{(n+1)}=\mat{A}\vect{s}^{(n)}$, for the $N\times N$ nonnegative matrix $\mat{A}$ such that
%\[
%A_{ij}=\left\{\begin{array}{l l}
%\ind{i=j}&\text{if $i\in R$}\\ \\
%\frac{ \displaystyle (\alpha \ind{j\in R}+(1-\alpha)\ind{j \notin R}) t_{ij}^2}{\displaystyle \alpha \sum\limits_{j_1 \in N_R(i)} t_{ij_1} + (1-\alpha) \displaystyle\sum\limits_{j_2 \in N_N(i)} t_{ij_2}}&\text{if $i>|R|$,}
%\end{array}\right.
%\]
%where $\ind{X}=1$ if $X$ is satisfied and $0$ otherwise.

Note that we take here the same weights as in~(\ref{eq:weights}), so that the notation above becomes
\begin{equation}\label{eq:proofnonrater}
i\notin N_R\quad\Rightarrow\quad s_i^{(n+1)}=\sum\limits_{j\in N(i)}w_{ij}s_j^{(n)}.
\end{equation}
It is easy to see that $\forall i \in N$, $0 \leq \displaystyle\sum_{j \in N(i)}w_{ij} \leq 1$, the right side becoming an equality only when a node trusts all its neighbours with a trust value of $1$.

We now %prove that the vector series $\vect{s}^{(n)}$ converges, showing
show by induction that for all nodes $i$, $s_{i}^{(n)}$ is in the interval $[0,1]$, is nondecreasing in $n$, and equals $0$ if $i$ has no trust path to any rater. 
\begin{itemize}
\item We immediately see from~(\ref{eq:s0}) and~(\ref{eq:series}) that $\vect{s}^{(0)}$ and $\vect{s}^{(1)}$ are in $[0,1]^N$, and that $\vect{s}^{(1)}\geq\vect{s}^{0}$. We also easily remark that if $i$ has no trust path to any rater, then $s_i^{(0)}=s_i^{(1)}=0$.
\item Now take $n_0 \geq 1$, and assume that for every $1 \leq n\leq n_0$, $\vect{s}^{(n)}\in[0,1]^N$, and $\vect{s}^{(n)}\geq\vect{s}^{(n-1)}$.

Each term of $\vect{s}^{(n_0+1)}$ is a weighted sum of the terms in $\vect{s}^{(n_0)}$ with nonnegative weights whose sum is in $[0,1]$, therefore $\vect{s}^{(n_0+1)}\in[0,1]^N$. 

Moreover, consider a node $i$ with no trust path to any rater. Then $s_i^{(n_0+1)}$ is a weighted sum of the scores $(s_j^{(n_0)})_{j\in N_N(i)}$ of the nodes $j$ that $i$ trusts. But none of those nodes has a trust path to any rater (otherwise we would have a trust path from $i$ to a rater), thus $s_j^{(n_0)}=0$, which implies $s_i^{(n_0+1)}=0$.

Finally, for every non-rater node $i$, we have from~(\ref{eq:proofnonrater})
\[
s_{i}^{(n_0+1)}-s_i^{(n_0)}=
%\frac{ \displaystyle\alpha\!\!\!\sum\limits_{j_1 \in N_R(i)}\!\!\! t_{ij_1}^2(s_{j_1}^{(n_0)}-s_{j_1}^{(n_0-1)}) + \displaystyle(1-\alpha)\!\!\!\sum\limits_{j_2 \in N_N(i)}\!\!\! t_{ij_2}^2 (s_{j_2}^{(n_0)}-s_{j_2}^{(n_0-1)})}{\displaystyle \alpha \!\!\!\sum\limits_{j_1 \in N_R(i)}\!\!\! t_{ij_1} + (1-\alpha)\!\!\!\sum\limits_{j_2 \in N_N(i)}\!\!\! t_{ij_2}}
\sum_{j\in N(i)}w_{ij}(s_{j}^{(n_0)}-s_{j}^{(n_0-1)})
\geq 0,
\]
where we used the induction hypothesis and the nonnegativity of the weights $w_{ij}$.

Since for each rater $j$, $s_j^{(n_0+1)}=s_j^{(n_0)}$, we have $\vect{s}^{(n_0+1)}\geq \vect{s}^{(n_0)}$.
\end{itemize}
The coordinates of $(\vect{s}^{(n)})_{n\in{\mathbb N}}$ are therefore nondecreasing and upper bounded by $1$. Consequently, $(\vect{s}^{(n)})_{n\in{\mathbb N}}$ converges to a vector $\vect{\bar s}\in[0,1]^N$, that satisfies Assumption~\ref{as:virgin_score}. From~(\ref{eq:series}) we see that $\vect{\bar s}$ also satisfies~(\ref{sat2d}).

%where $\ind{X}=1$ if $X$ is satisfied and $0$ otherwise.

%\color{violet}
%Let us split the set $N_R$ of non-raters into
%\begin{eqnarray*}
%N_R^+&\define&\{i\in N_R: \text{there is a trust path from $i$ to a rater}\}\\
%\text{and }N_R^-&\define&N_R\setminus N_R^+.
%\end{eqnarray*}
%
%
%Without loss of generality, we re-number the nodes such that the first $|R|$ nodes are all raters, nodes $|R|+1$ to $|R|+|N_R^+|$ are in $N_R^+$, and the last ones are nodes in $N_R^-$. Then the system~(\ref{sat2d}) is equivalent to the matrix equation
%\begin{equation}\label{eq:mat_eq}
%\vect{s}=\mat{A}\vect{s},
%\end{equation}
%where $\mat{A}$ is the $N\times N$ nonnegative matrix such that
%\[
%A_{ij}=\left\{\begin{array}{r l}
%\ind{i=j}&\text{if $i\leq |R|$ (i.e., $i$ is a rater)}\\ \\
%\frac{ \displaystyle (\alpha \ind{j\leq |R|}+(1-\alpha)\ind{j>|R|}) t_{ij}^2}{\displaystyle \alpha \sum\limits_{j_1 \in N_R(i)} t_{ij_1} + (1-\alpha) \displaystyle\sum\limits_{j_2 \in N_N(i)} t_{ij_2}}&\text{if $i>|R|$,}
%\end{array}\right.
%\]
%where $\ind{X}=1$ if $X$ is satisfied and $0$ otherwise.

{\bf Uniqueness.} Assume that there exist two distinct vectors $\vect{\bar s}$ and $\vect{\tilde s}$ satisfying both the set of equations~(\ref{sat2d}) and Assumption~\ref{as:virgin_score}. Without loss of generality, we can then assume that there is a node $i\in V$ such that $\bar s_i>\tilde s_i$. 

Define $i_0\define \arg\max_i\bar s_i-\tilde s_i$. We can immediately say that node $i_0$ must have a trust path to a rater (from Assumption~\ref{as:virgin_score}), but is not a rater itself (from~(\ref{sat2d})). Then, from~(\ref{eq:nonrater}) we should also have
\[
\bar s_{i_0}-\tilde s_{i_0}=\sum_{j\in N(i_0)}w_{ij}(\bar s_j-\tilde s_j)\leq (\bar s_{i_0}-\tilde s_{i_0})\sum_{j\in N(i_0)}w_{ij}\leq \bar s_{i_0}-\tilde s_{i_0},
\]
where the first inequality comes from the definition of $i_0$ and the positivity of $(w_{ij})_{j\in N(i_0)}$, and the last one from the fact that $\sum_{j\in N(i_0)} w_{ij}\leq 1$. 
We then obtain that these inequalities are actually equalities, and since the weights $(w_{ij})_{j\in N(i_0)}$ are strictly positive this implies that $\sum_{j\in N(i_0)} w_{ij}= 1$, and that
\[
j\in N(i_0)\quad\Rightarrow\quad \bar s_j-\tilde s_j=\bar s_{i_0}-\tilde s_{i_0}.
\]
In particular, this means that no node in $N(i_0)$ is a rater (otherwise we would have a contradiction with~(\ref{sat2d}) as $\bar s_j-\tilde s_j$ would equal zero for the rater and $\sum_{j\in N_N(i_0)}w_{ij} < 1$).

But we can apply the exact same reasoning that we used for $i_0$ to all nodes in $N(i_0)$, and recursively show that $i_0$ does not trust (even indirectly) any rater, a contradiction. Consequently there can be only one solution to~(\ref{sat2d}) satisfying Assumption~\ref{as:virgin_score}.
\end{proof}

\subsection{Algorithm to calculate satisfaction}

In the previous section, we observed that the satisfaction scores of non-raters can be obtained as a solution of a system of $N-|N_R|$ linear equations. We argue why this approach is not computationally efficient when the number of users in the community is large, and present a more efficient iterative algorithm that generates solutions up to desired degrees of accuracy. 

Commonly used computer packages and toolboxes solve systems of linear equations by some variant of Gaussian Elimination~\cite[Chapter~2]{moler2004numerical}. The standard Gaussian Elimination uses a series of row operations to convert a given system $Ax=b$ to the form $Ux=y$ ($LU$ decomposition) where $U$ is an upper triangular matrix, and then solves the latter by backward substitution. It is well known that while Gaussian Elimination requires $O(N^3)$ mathematical operations, its actual time complexity is much worse even with the most efficient of implementations~\cite{fang1997worst}.

Most of today's social networks boast of large registered user bases, and this number is only expected to grow with the proliferation of web-based technologies and high speed Internet. Hence, the social graph can contain anywhere between $1000$ and million nodes, and it is not feasible to execute even a $O(N^3)$ algorithm, let alone complexities worse than cubic which is normally the case. Moreover, the memory space required is also of the order of $O(N^2)$ or greater. This implies that any mechanism which has to repeatedly calculate satisfaction scores every time a non-rater becomes a rater requires a computationally efficient algorithm that uses the structure of the graph and existing satisfaction scores to estimate new satisfaction values. Such an algorithm, which effectively exploits the sparse nature of these graphs is proposed here. 

In order to maintain consistency of notation we use $(\vect{s}^{(n)})_{n\in{\mathbb N}}$ to represent the column vector of all satisfaction scores at the $n^{th}$ iteration of the satisfaction computation algorithm and ${s}^{(n)}_{i}$ for the satisfaction score of the $i^{th}$ user after the $n^{th}$ iteration. The  algorithm to be discussed uses the following $N\times N$ matrix $\mathbf{A}$:
\[
       A_{ij}= \left\{ \begin{array}{ll}
   1 &\mbox{if $i=j$} \\
  0 &\mbox{if $i \neq j$}       \end{array}\right.
           \mbox{$i$ is a rater}
 \]
   
\[
  A_{ij} =
    \left\{ \begin{array}{cl}
   w_{ij} &\mbox{if $i \neq j$} \\
  0 &\mbox{if $i=j$}   
            \end{array}\right.  \mbox{$i$ is a non-rater}
 \]

Based on the above matrix, we define our iterative algorithm as:
\begin{algorithm}
\caption{Iterative Satisfaction Algorithm}
\label{iteralgo}
\algsetup{indent=2em}
\begin{algorithmic}[1]
\STATE $T \leftarrow$ Maximum Tolerance/Error
\STATE $M \leftarrow$ Maximum Number of Iterations
\STATE $n \leftarrow 0$
\STATE $\vect{s}^{(n)} \leftarrow \vect{s}_{prev} \mbox{ or } \mathbf{e_0} \quad \quad \quad$ \COMMENT{\textit{$\mathbf{e_0}$ represents a column vector of all zeros}}
\WHILE{$\exists j, (\vect{s}^{(n)}_{j} - \vect{s}^{(n-1)}_{j}) > T \quad  \AND \quad n < M$}
\STATE $n \gets n + 1$
\STATE $\vect{s}^{(n)} = \mathbf{A}.\vect{s}^{(n-1)}$
\ENDWHILE
\end{algorithmic}
\end{algorithm}\\
The algorithm uses matrix-vector multiplications in order to arrive upon the desired vector of satisfaction values, which (approximately) satisfy Equation~(\ref{sat2d}). Normal implementations use $O(N^2)$ operations for each matrix-vector product and thus, the overall efficiency of the algorithm is $O(MN^2)$. As long as $M < N$, Algorithm~(\ref{iteralgo}) provides a distinct advantage over directly solving the linear equations. However, in general, the underlying social graph of the community is very sparse as each user trusts only a limited number of other users, especially when the total number of users who have access to the document is very large, i.e. $\forall i, |N(i)| \ll N$. In such cases, there are very efficient representations for sparse matrices like Compressed Row Storage (CRS)~\cite{barrett1994templates} which use space only in the order of $O(n_{nz} + N)$, where $n_{nz}$ is the total number of non-zero entries in the trust matrix or in other words, the number of edges in the social graph($|E|$). CRS requires only $O(n_{nz})$ operations for matrix-vector multiplications, which is a significant increase in efficiency over $O(N^2)$ as long as the matrix is sparse.

It is easy to see that the convergence of Algorithm~\ref{iteralgo} follows from the Existence proof given in Section~(\ref{sec:exist}). The $(\vect{s}^{(n)})_{n\in{\mathbb N}}	$ vectors obtained after every iteration of the above algorithm are the same as the vectors defined in Equations~(\ref{eq:s0}) and~(\ref{eq:series}). We conclude that $\vect{s}^{(n)}$ converges to a unique vector $\vect{\bar{s}} \in [0,1]^N$ that satisfies Equation~(\ref{sat2d}).

\section{Analytical Results in Random Graphs}\label{sec:random_graphs}
\label{analytic}

The Erd\"{o}s-R{\'e}nyi model of Random Graphs is an interesting, yet simple mathematical tool that allows us to analyze social and technological networks. In this model, every possible edge in a graph has a certain probability of being present, and is independent of every other edge. We denote by $G(N,p)$, a directed random graph with $N$ nodes where $p$ is the probability that any given edge $(i,j)$ is present. By straightforward reasoning, the expected number of edges in the graph is ${pN(N-1)}$. Since each edge can be viewed as a Bernoulli random variable with probability of success $p$, the degree of a node follows the Binomial distribution. If we define a parameter $\lambda$, such that $p= \sfrac{\lambda}{N}$, then the expected degree of a node comes out to be $\approx \lambda$. Keeping the expected degree constant as $N \rightarrow \infty$, the node degree $D$ becomes a Poisson random variable with parameter $\lambda$. 

\begin{equation}
P(D=d) = e^{-\lambda} \frac{\lambda^d}{d!}
\label{poisson}
\end{equation}

The behavior of random graphs as $\lambda$ varies has been well established: when $\lambda < 1$, the connected components are small and no larger than $O(\log N)$; when $\lambda > 1$, a giant connected component of size $O(N)$ emerges. For a thorough treatment of random graphs, the reader is advised to go through the text by B\'{e}la Bollob{\'a}s~\cite{bollobas2001random}.

In this section, we focus on the choice of the set of raters. Our goal is to compare the performance (in terms of proportions of satisfied users) of different rater sets, and possibly minimize the total effort needed to validate a document. In order to reach analytical results in that direction, we consider here a specific type of instance, satisfying the assumptions below.

\begin{assumption}\label{as:erdos}
The trust relationship graph $(V,E)$ is an Erd\"{o}s-R{\'e}nyi graph with parameter $\lambda$, and the number of nodes tending to infinity.
\end{assumption}

\begin{assumption}\label{as:analytical_simplif}
The trust values and the rater behavior are such that:
\begin{itemize}
\item for any edge $(i,j)\in E$, the trust value is the same and is denoted by $t_{i,j}= t$,
\item all raters evaluate the document with the same rating $r>0$.
\end{itemize}
\end{assumption}
Considering trust values $t_{i,j}$ (resp. rater evaluations $r_i$) to be the same for all pairs of connected nodes (resp., raters) is a very strong assumption. We however make that choice in order to highlight the influence of the structure of the social network, that is the topology of the graph $(V,E)$, on the resulting satisfaction values of users. In the end, we wish to exploit the topology to optimize the selection of the raters (e.g., to minimize the total evaluation effort over the network). Having heterogeneous trust values would bring some added richness to the model, which should be taken into account in the rater selection decision, an interesting extension of the scheme (greedy algorithm) we develop here. On the other hand, relaxing the assumption of all raters giving the same score seems less interesting, since in practice the score set by a rater cannot be predicted (and thus, considered as an input for the rater selection problem): due to the linearity of the satisfaction calculation model, we expect that assuming the scores to be randomly distributed with mean value $r$ would yield similar results. 

\subsection{Distribution of satisfaction scores with uniformly selected raters}

%We consider a directed random graph with $N$ nodes (we assume that $N$ is large), and each of the potential $N(N-1)$ edges is present with probability $\lambda/N$, for some fixed value $\lambda$. Further, we assume that for all edges $(k,l)$ in the graph, we have $t_{k,l}=t$. 
Suppose that Assumptions~\ref{as:erdos} and~\ref{as:analytical_simplif} hold.
We are interested in calculating the (cumulative) distribution $F$ of satisfaction scores among non-rating nodes, given that a fraction $k$ of nodes are raters, picked uniformly among the $N$ users. In other words, when we pick a non-rater, $F(x)$ is the probability that his satisfaction score is below or equal to $x$, for $x\in[0,1]$. By definition, $F(0) = c$, where $c$ is the fraction of users in the connected components with no raters, and $F(1)=1$.
The symmetry inherent in random graphs and the fact that every node has an equal probability of being a rater ($k = \sfrac{|N_R|}{N}$) ensures that the satisfaction scores of non-raters can be treated as independent and identically distributed random variables.% As far as raters are concerned, we assume that they uniformly rate the document with a rating $r$.

We start progressively, considering $d$ the number of nodes that a non-rater trusts, and computing $F_d(x)$, the probability that the satisfaction score of the node is smaller than $x$ \emph{given that the node has $d$ trustees}:
\begin{itemize}
\item If $d=0$, then the satisfaction score of the node is $0$, and $F_0(x)=1, \forall x$.
\item If $d=1$, then the satisfaction score is the one of the neighbor, multiplied by the trust value $t$. That neighbor is:
\begin{itemize}
\item a rater with probability $k$ (thus, with satisfaction score $r$ under Assumption~\ref{as:analytical_simplif}). Then the probability of the node satisfaction being below $x$ is $\ind{tr\leq x}$, where $\ind{A}$ equals $1$ if condition $A$ is satisfied and $0$ otherwise.
\item a non-rater with probability $1-k$ (thus, with a satisfaction score distributed according to $F$). Then the satisfaction score of our node is below $x$ if and only if the satisfaction score of the trustee is below $x/t$, which occurs with probability $F(\min(x/t,1))$.
\end{itemize}
As a result, if $d=1$ the expected probability that the satisfaction score of the node is below $x$ is
\[
F_1(x)=k\ind{tr\leq x}+(1-k)F(\min(x/t,1)).
\]
\item If $d\geq 2$, we use conditional probabilities as in the previous case, varying the number $\ell$ of neighbors who are raters: the probability that there are $\ell$ raters among the $d$ neighbors is $C_\ell ^dk^\ell(1-k)^{d-\ell}$, with $C_\ell ^d=\frac{d!}{\ell !(d-\ell)!}$. If we denote by $f$ the (unknown) probability density function of $F$, then we have:
\newcommand{\points}{\!\cdot\!\cdot\!\cdot\!}
\begin{eqnarray}
F_d(x)&\!\!\!=\!\!\!&\sum_{\ell=0}^dC_\ell ^d k^\ell(1-k)^{d-\ell}{\mathbb P}(\sum_{\text{non-rater neighbors }i}s_i\leq \mu x - cr)\nonumber\\
&\!\!\!=\!\!\!&\sum_{\ell=0}^dC_\ell ^d k^\ell(1-k)^{d-\ell}\!\!\!\int_{s_1\geq 0
}\!\!\!\!\!\points\int_{s_{d-\ell}\geq 0
}\!\!\!\!\!\!\!f(s_1)
\points f(s_{d-\ell})\ind{\sum_{i=1}^{d-\ell}s_i\leq \mu x - cr} ds_1
\points ds_{d-\ell},\label{fixedpt}
\end{eqnarray}
where $\mu = \frac{1}{t} \{c + (d-\ell)\}$ and $c = \frac{\alpha \ell}{(1-\alpha)}$. 
\end{itemize}

Finally, the fixed point equation that the distribution $F$ (or the density $f$) of satisfaction scores among non-raters should satisfy is the following:\vspace{3ex}

\noindent $\forall x\in[0,1]$,
\begin{equation}
F(x)=e^{-\lambda}\sum_{d=0}^{\infty}\frac{\lambda^d}{d!}F_d(x),
\label{eq:fixed_point}
\end{equation}
where $F_d(x)$ is given in~(\ref{fixedpt}).

\subsection{Bound on the number of raters needed to validate a document}
In order to prove that the model defined in Section~\ref{model} results in improved efficiency, we need to show bounds on the number of raters required to satisfy a sufficient proportion of the collaborators. In other words, we argue that the Satisfaction Estimation Model leads to reduced human effort with respect to a situation where each user should evaluate the document separately to validate it. We show here an upper bound on the maximum number of unsatisfied users in the system, depending on the proportion $k$ of raters. The result is derived by using the reverse Markov inequality on the expected satisfaction value of users. In Section~\ref{sec:simulations}, we show the actual number of raters required to pass the document, obtained via extensive simulations in various network conditions.

We first claim that for any instance of $G_R=((V,E),T,B,N_R,R)$, $\forall i \in V$ and any value $\tilde\alpha > 0.5$, we have: 
 \begin{equation}
 s_{i,\tilde\alpha} \geq  s_{i,(\alpha =0.5)},
\label{alphacomparison}
 \end{equation}
where $s_{i,\alpha}$ represents the satisfaction score of user $i$ that is obtained when users apply a weight $\alpha$ to the scores of raters in their satisfaction computation~(\ref{sat2d}).
 
This means that the satisfaction score of a non-rater is always greater when raters are given more weightage than when equal weight is given to the satisfaction of raters and non-raters. This can be proved inductively by using an approach similar to the proof of existence given in Section~\ref{sec:exist} and using as the induction hypothesis:
\[
\frac{\tilde\alpha \ell r + (1-\tilde\alpha)\sum_{j=1}^{d-\ell}s_{j,\tilde\alpha}^{(n)}}{\tilde\alpha \ell + (1-\tilde\alpha)(d-\ell)} \geq \frac{\alpha \ell r + (1-\alpha)\sum_{j=1}^{d-\ell}s_{j,\alpha}^{(n)}}{\alpha  \ell + (1-\alpha)(d-\ell)}
\]
in the case when $\alpha = 0.5$, and $d,\ell,r$ are the number of users that the non-rater trusts, the number of raters among those, and the (common) score set by raters respectively, as defined previously.  We now derive the expected satisfaction score of a non-rater as a function of the proportion $k$ of raters, when $\alpha=0.5$. The expected satisfaction score when $\alpha > 0.5$ is always greater than this value and hence the bounds obtained hold for all cases. 

\begin{proposition}\label{prop:exp_satisf}
Consider Assumptions~\ref{as:virgin_score},~\ref{as:erdos} and~\ref{as:analytical_simplif} hold, and that a proportion $k$ of raters is randomly selected among the set of users, according to a uniform law. Then if users weigh non-raters' opinions as much as raters' scores (i.e. $\alpha=0.5$), the expected satisfaction score among non-raters is
\begin{equation}
{\mathbb E}(s_i)=  \frac{trk(1-e^{-\lambda})}{\left( 1-t(1-k)(1-e^{-\lambda}) \right)}.
\label{expectedsat}
\end{equation}
\end{proposition}

\begin{proof}
Denoting by $D$ and $L$ the random variables giving the number of nodes trusted by $i$ (degree of node $i$) and the number of raters among those trustees, respectively, we have
\begin{eqnarray*}
{\mathbb E}(s_i)  &=& \sum_{d=0}^{N-1} {\mathbb E}(s_i | D=d) {\mathbb P}(D = d)\\
&=& \sum_{d=1}^{N-1} {\mathbb E}(s_i | D=d) {\mathbb P}(D = d)\\
&=& \sum_{d=1}^{N-1}\left\lbrace \sum_{l=0}^{d} {\mathbb E}(s_i | L=\ell, D=d) {\mathbb P}(L=\ell |d)\right\rbrace {\mathbb P}(D = d)\\
\end{eqnarray*}
where ${\mathbb P}(A|B)$ refers to the probability of $A$ conditional on $B$% has already occurred
, and the second line comes from Assumption~\ref{as:virgin_score}.

Now, for $d\geq 1$ we have:
\[
{\mathbb E}(s_i | L=\ell, D=d) = \frac{t}{d} {\mathbb E}\left[ \ell r +  \displaystyle\sum\limits_{j \in N_N(i)} s_j \right]
\]
By linearity of expectation, and using the notation $s\define {\mathbb E}(s_i)$, we have
\begin{eqnarray}
{\mathbb E}(s_i | \ell, d)& =& \frac{t}{d} (\ell r +  (d-\ell)s ) \nonumber \\
 {\mathbb E}(s_i | d)& =& \frac{t}{d} \sum_{\ell =0}^d C_{\ell}^d k^\ell (1-k)^{d-\ell} (\ell r +  (d- \ell)s )  \nonumber \\
 &=& \frac{t}{d} \left(rdk + sd (1-k)\right) \nonumber \quad \mbox{(using that $\sum_{\ell =0}^d \ell C_{\ell}^d k^\ell (1-k)^{d-\ell}={\mathbb E}[L|D=d] = dk$)} \\
 & = & t\left(rk + s(1-k)\right) \nonumber 
\end{eqnarray}
Interestingly, this is independent of $d$ (when $d\geq 1$). Two cases therefore arise,\\
\textbf{Case I:} $d=0$
$${\mathbb E}(s_i) = 0 \quad \mbox{(by Assumption~\ref{as:virgin_score})}$$
$${\mathbb P}(D=0) = e^{-\lambda}$$
\textbf{Case II:} $d\geq 1$
$${\mathbb E}(s_i) = t(rk + s(1-k))$$
$${\mathbb P}(D=d) = e^{-\lambda}\frac{\lambda^d}{d!}$$
Hence, we have
\begin{eqnarray}
 {\mathbb E}(s_i)& =& t \sum_{d =1}^{\infty} e^{-\lambda}\frac{\lambda^d}{d!} (rk + s(1-k) )  \nonumber \\
s &=& t (1- e^{-\lambda}) (rk + s(1-k) ) \nonumber \\
\implies s\left\lbrace 1-t(1-k)(1-e^{-\lambda}) \right\rbrace & = & trk(1-e^{-\lambda}), \nonumber
\end{eqnarray}
which establishes the proposition.
\end{proof}

%\subsubsection*{Number of unsatisfied users}
We can now use this result to provide bounds on the number of raters that are needed so that a target proportion of the community is satisfied with the document.
\begin{proposition}\label{prop:bound}
Consider Assumptions~\ref{as:virgin_score},~\ref{as:erdos} and~\ref{as:analytical_simplif} hold, and assume that the minimum quality requirements of users are %identically distributed, and independent of all other variables. 
all equal to some value $b\in(0,1)$ such that $r > b$.
Then if raters are picked uniformly among the set of users, for a proportion $T$ of the non-raters to be satisfied it is necessary that the proportion of raters be at least
%\[(1-T)b\left( 1-t(1-k)(1-e^{-\lambda}) \right)  \geq b-t(1-e^{-\lambda})(rk + b(1-k))\]
%\[(1-T)b\left( 1-t(1-e^{-\lambda})+kt(1-e^{-\lambda}) \right)  \geq b-t(1-e^{-\lambda})(k(r-b) + b)\]
%\[(1-T)b\left( 1-t(1-e^{-\lambda})\right)+kt(1-T)b(1-e^{-\lambda})  \geq b(1-t(1-e^{-\lambda}))-kt(1-e^{-\lambda})(r-b)\]
%\[kt(1-e^{-\lambda})\left[(1-T)b+(r-b)\right]  \geq Tb(1-t(1-e^{-\lambda}))\]
%\[k \geq \frac{Tb(1-t(1-e^{-\lambda}))}{t(1-e^{-\lambda})\left[(1-T)b+(r-b)\right] }\]
%
\begin{equation}\label{eq:bound2}
k^{\min}(T)=\frac{b\tilde T(1-t(1-e^{-\lambda}))}{t(1-e^{-\lambda})(r-b\tilde T)},
\end{equation}
while it is sufficient to choose a proportion of raters
%
%\pat{
%\[1-T  \geq \frac{1-t(1-e^{-\lambda})(1-k(1-r))}{(1-b)\left( 1-t(1-k)(1-e^{-\lambda}) \right)}\]
%[(1-T)(1-b)\left( 1-t(1-e^{-\lambda})+k(1-e^{-\lambda}) \right)  \geq 1-t(1-e^{-\lambda})+kt(1-r)(1-e^{-\lambda})\]
%\[(1-T)(1-b)\left( 1-t(1-e^{-\lambda})\right)+k(1-T)(1-b)(1-e^{-\lambda})   \geq 1-t(1-e^{-\lambda})+kt(1-r)(1-e^{-\lambda})\]
%\[\left( 1-t(1-e^{-\lambda})\right)\left((1-T)(1-b)-1\right)\geq k(1-e^{-\lambda})\left[t(1-r)-(1-T)(1-b)\right]\]
%\[\left( 1-t(1-e^{-\lambda})\right)\left((T+b)-Tb\right)\leq k(1-e^{-\lambda})\left[t(1-r)-(1-T)(1-b)\right]\]
%\[k^{\max}(T)=\min\left(\frac{\left( 1-t(1-e^{-\lambda})\right)\left((T+b)-Tb\right)}{(1-e^{-\lambda})\left[t(1-r)-(1-T)(1-b)\right]},\quad 1\right)\]
%}
%
\begin{equation}\label{eq:bound1}
k^{\max}(T)=\min\left(\frac{\left(1-t(1-e^{-\lambda})\right)(T(1-b) + b)}{t(1-e^{-\lambda})\left(r-b - T(1-b)\right)},\quad 1 \right).
\end{equation}
Consequently, to satisfy a proportion $\tilde T$ of the whole community:
\begin{itemize}
\item it is necessary that the proportion of raters be at least $k^{\min}(T^{\min})$, with $T^{\min}$ the solution of the equation $(1-k^{min}(T))(1-T)=1-\tilde T$,
\item it is sufficient that the proportion of raters be at least $k^{\max}(T^{\max})$, with $T^{\max}$ the solution of the equation $(1-k^{max}(T))(1-T)=1-\tilde T$.
\end{itemize}

\end{proposition}

\begin{proof}
We can  apply Markov's inequality to obtain a lower bound on the number of unsatisfied users.
\[
{\mathbb P}(s_i\leq b) \geq \frac{b-{\mathbb E}(s_i)}{b}
\]

Substituting~(\ref{expectedsat}) in this equation, we get
\begin{equation}
{\mathbb P}(s_i \leq b)  \geq \frac{b-t(1-e^{-\lambda})(rk + b(1-k))}{b\left( 1-t(1-k)(1-e^{-\lambda}) \right)},
\label{mainbound2}
\end{equation}
which leads to~(\ref{eq:bound2}) once we write the target condition ${\mathbb P}(s_i \leq b)\leq 1-T$.

Again applying the Markov inequality, we obtain
\[
{\mathbb P}(s_i\leq b) ={\mathbb P}(1-s_i\geq 1-b) \leq \frac{1-{\mathbb E}(s_i)}{1-b}.
\]
Then substituting~(\ref{expectedsat}) in the above equation,we get
\begin{equation}
{\mathbb P}(s_i \leq b)  \leq \frac{1-t(1-e^{-\lambda})(1-k(1-r))}{(1-b)\left( 1-t(1-k)(1-e^{-\lambda}) \right)},
\label{mainbound}
\end{equation}
which yields~(\ref{eq:bound1}).
%This provides an upper bound on the number of unsatisfied users in the system given the fraction of raters.	It is important to note that this bound is applicable for all values of $\alpha \geq 0.5$, albeit not so tight, since using~(\ref{alphacomparison}), we can show
%\[
%(s_{i,\alpha} \leq b) \leq \frac{1-E(s_{i,\alpha})}{1-b} \leq \frac{1-E(s_{i})}{1-b}
%\]

Now, since $r>b$ all raters are satisfied, so a proportion $T$ of non-raters being satisfied corresponds to a proportion $k+(1-k)T$ of satisfied users (or equivalently, a proportion $(1-k)(1-T)$ of unsatisfied users) over the whole community. This implies the second part of the proposition.
\end{proof}

\section{Selecting raters to maximize satisfaction}\label{sec:select_raters}
In this section, we consider an optimization problem that is closely connected to our model, that of finding the optimal set of raters (given a maximum number of raters) in order to maximize satisfaction. We show that this problem is NP-Hard irrespective of the chosen parameters. An approximation algorithm based on the concept of marginal cost and approximate oracles is then proposed and the improvement in efficiency obtained due to this algorithm is shown via simulations in the next section.

\subsection{Complexity of Optimization Problem}
We define the optimization problem of selecting the best possible raters as follows:

\textbf{MAXIMUM-SATISFACTION} Given:
\begin{itemize} 
\item a social network represented by $G_{\alpha}\define((V,E),T,B,r,\alpha)$, where $r$ is the rating that any rater would give to the article under consideration and all other parameters are as defined previously,
\item and an integer $\kappa$ that denotes the maximum cardinality of the rater set ($|N_R| \leq \kappa$),
\end{itemize}
find the set $N_R\subset V$ of maximum size $\kappa$ that satisfies the maximum number of users in the network according to Equation~(\ref{model}) and Assumption~\ref{as:virgin_score}.%, i.e. if $S$ is the set of users whose satisfaction scores are greater than their thresholds, our objective is to maximize $|S|$.

We now prove that this rater selection problem is computationally hard.
\begin{proposition}
\textbf{MAXIMUM-SATISFACTION} is NP-Hard (even under the simplifying assumptions~\ref{as:analytical_simplif}).
\end{proposition}
\begin{proof}
The problem that we use for our reduction is the NP-Hard problem Maximum k-Cover, also known as the Maximum Coverage Problem. The problem is stated as follows: 
Given universe $U=\{U_1,U_2,\cdots ,U_m\}$ and a collection of subsets $S=\{S_1,S_2,\cdots ,S_n\}$ such that $S_i \subset U, \forall i$, find a collection of subsets $S^* \subset S$ which maximizes $|\cup _i S^*_i|$ and satisfies $|S^*| \leq \kappa$. In other words, we have to cover as many elements as possible in the universe. We assume that $(\cup _i S_i) = U$. We now reduce this to an instance of \textbf{MAXIMUM-SATISFACTION}.

Consider a bipartite graph $G_{\alpha}=((V,E),B,T,r,\alpha)$ such that $|V|=m+n$. The $m+n$ nodes are $(U_1,U_2,\cdots ,U_m)$ and $(S_1,S_2,\cdots ,S_n)$. There is a directed edge $U_j \rightarrow S_i$ iff $U_j \in S_i$. Let $t_{ij} = t$ for all edges. All of $t$, $\alpha$ and $r$ lie in the range $(0,1]$. Let $b_i =  0$ $\forall i \in V$. The purpose of setting all the thresholds to zero is to ensure that a user would be satisfied as soon as his satisfaction score becomes non-zero, as $s_i > b_i$ is the condition for satisfaction.

\textbf{Claim} Any algorithm for Maximum Satisfaction of $(G_\alpha,\kappa)$ also yields a solution for Maximum k-cover.

\begin{lemma} Any algorithm that selects raters to maximize satisfaction only selects nodes from $S$.
\end{lemma}
Suppose an optimal solution contains the node $U_j \in U$, we can replace $U_j$ by $S_i \in S$ such that $(U_j,S_i)\in E$ and the solution still remains optimal. Such an $S_i$ exists because $(\cup_k S_k) = U$. In the case when $S_i$ already belongs to the set $S^*$, then the solution is improved by one which contradicts our claim that the previous solution is optimal. Therefore, there exists an optimal solution with nodes only from $S$, which the algorithm finds.

$\forall U_j$, $s_{U_j} > 0$ iff the optimal solution contains a set $S_i$, which has an incoming edge from $U_j$, i.e. $U_j$ trusts $S_i$. Thus any algorithm for Maximum Satisfaction, selects $\kappa$ nodes from $S$ such that the number of satisfied nodes in $U$ is maximized. This completes our proof that the \textbf{MAXIMUM-SATISFACTION} problem is NP-Hard.
\end{proof}

\subsection{Greedy Algorithm to Select Raters}
\label{maxsatisfaction}
In the preceding sections, we proposed a model for estimating satisfaction scores of users in a collaborative network given the ratings of users who have already read the document. Based on this model, we described the complete collaborative editing process. However, we have so far avoided specifying any particular method to select raters apart from randomization. Although we proved that choosing the best possible users to review a document is NP-Hard,  it seems likely that selecting raters based on the trust network (in polynomial time), may be more efficient than a uniform random choice. Here, we propose an algorithm that greedily selects raters who are likely to satisfy the maximum number of non-raters at each stage of the document development. In the following section, we show that this algorithm outperforms the other greedy algorithms based on trust, and also the random rater selection strategy.
 
Quite a few greedy algorithms present themselves as viable solutions once the social network structure ($G_R=((V,E),T,B,N_R,R)$) is obtained and the weight $\alpha$ is fixed. On the surface, it seems like a good idea to select the non-rater on whom other non-raters place the maximum amount of trust, to read and rate the article, i.e. we choose $i$ such that $i \notin N_R$ and $\sum_{j \notin N_R} t_{ji}$ is maximum over all $i$. However, it appears that this algorithm is short-sighted in the sense that it only takes into account the satisfaction of the immediate neighbors without considering the impact of the rater as the rating propagates across the network. We later show that this is indeed the case. Therefore, we aim for a greedy algorithm that selects the user who has the maximum impact on all users across the network.

Before presenting our algorithm, it is useful to dwell upon the concept of Marginal Cost, which we borrow from economics. \textit{Marginal cost} $(\text{MC})$ is a concept that is used to indicate the change in total cost that arises when the quantity produced changes by one unit. That is, it is the cost of producing one extra unit of a good. If the total cost function $\text{TC}$ is differentiable, this means that 
\[ \text{MC} = \frac{\,d (\text{TC})}{\,dQ}, \]
where $Q$ is the independent variable denoting the quantity being produced. Typically, $\text{TC}$ may be a linear or nonlinear function of $Q.$ It is not uncommon in real problems for the 
variable $Q$ to take only discrete values. In such a case the marginal cost is redefined in a macroscopic sense, that is, 
\[ \text{MC} = \frac{\Delta (\text{TC})}{\Delta Q}. \]

Let $U: G_R \rightarrow \mathbb{N}$ be a function that maps any particular state of a social network to the number of satisfied users in that network, i.e. $U(G) = |S|$, where $S$ is the set of users in the system whose satisfaction scores are greater than their thresholds. Then, we define a quantity \textit{Marginal Satisfaction} (MS) associated with each non-rater to be the number of non-raters newly satisfied when that particular user is chosen to rate the document. Mathematically,
\begin{equation}
MS_i = f(G_R,i) = U(G_R+\{i\}) - U(G_R).
\end{equation}
It is important to note the abuse of notation: all operations on $G_R$ are performed only on the set of raters i.e. $G_R + \{i\}$ denotes the same social network as $G_R$ except that user $i$ is now a rater. We assume that all raters provide the document a rating of $r$, or that we have the same a priori information about the scores that raters will provide, so that $r$ is the expected value of those scores. The value of $r$ can be increased at each successive stage of document editing to model the increasing quality of the document. Our greedy algorithm works as follows:

\begin{algorithm}
\caption{Greedy Algorithm based on Marginal Satisfaction to Select Raters}
\label{greedyalg}
\algsetup{indent=2em}
\begin{algorithmic}[1]
\WHILE{$\exists j, s_j\leq b_j \quad$}
\STATE  $\forall i \notin N_R$, \textbf{calculate} $f(G_R,i)$
\STATE \textbf{select} user $i^*$ such that $f(G_R,i^*) \geq f(G_R,i) \quad \forall i$
\STATE $N_R \rightarrow N_R \cup i^*$
\STATE \textbf{update} $s_j\  \forall j$ using Equation~(\ref{sat2d})
\ENDWHILE
\end{algorithmic}
\end{algorithm}
Basically at each stage, the algorithm chooses the non-rater with the maximum marginal satisfaction. We remark that this user may already be satisfied. The marginal satisfaction for each user can be %calculated
approximated in polynomial time by running Algorithm~\ref{iteralgo} for a fixed number of iterations. Thus, the iterative algorithm proposed serves as an approximate oracle. The only disadvantage of this approach when compared to the random selection of raters or other greedy algorithms is the added complexity in calculating Marginal Satisfaction for all users at every stage.

However, in the case when $\alpha=0.5$, we can exploit the linearity of our model and the fact that the weights in~(\ref{eq:weights}) remain constant as we increase the number of raters, and use a dynamic programming based approach to calculate the marginal satisfaction of each user. For this, we define a quantity  $\Delta_{N_R}(i,j)$ to be the increase in the satisfaction of non-rater $j$, when $i$ (also a non-rater) becomes a rater and increases his satisfaction by $1$, assuming all other parameters remain fixed. Then, our algorithm is motivated by the following observation: if user $k$ becomes a rater, then all trust paths between $i$ and $j$ that pass through $k$ cannot be considered anymore while calculating $\Delta_{N_R}(i,j)$. Therefore we have, 
\[\Delta_{(N_R \cup k)}(i,j) = \Delta_{N_R}(i,j) - \Delta_{N_R}(i,k)\times\Delta_{N_R}(k,j).\] In other words, due to the linearity of satisfaction scores, the impact that user $i$ has on $j$ considering only paths that pass through $k$ equals the impact of $i$ on $k$ times the impact of $k$ on $j$. Using this quantity, we can calculate the vector containing the increase in satisfaction values of all non-raters when user $i$ becomes a rater as $(r-s_i)\times\Delta_{N_R}(i)$, where $r$ is the rating as defined previously, $s_i$ denotes the current satisfaction of $i$ and $\Delta_{N_R}(i)$ is a vector consisting of $(\Delta_{N_R}(i,j))_{j\notin N_R}$. Adding the above vector to the vector of satisfaction scores, one can easily calculate the marginal satisfaction of user $i$. Therefore, at each stage one has to only maintain $\Delta_{N_R}(i,j)$ for all $i,j$ belonging to the set of non-raters, in order to calculate the marginal satisfaction of all users when $\alpha=0.5$. This can be done in polynomial time and is much more efficient than running Algorithm~(\ref{iteralgo}) for all users at each stage. In the following section, we use this method to compare the performance of our greedy algorithm to other algorithms.

\section{Simulation and Results}\label{sec:simulations}
\label{simulations}
In Section~\ref{analytic}, we gave a fixed point equation to solve for the distribution of satisfaction scores and showed bounds on the number of users required to review an article in random graphs to satisfy a given proportion of users. In this section, we are interested in examining the performance of our proposed model and greedy algorithm in practice, namely collaborative editing systems with large number of users. We study the effect of several network parameters including the number of raters and the density of the underlying graph on user satisfaction. We first show that our satisfaction estimation model results in a considerable efficiency gain with respect to the case when all collaborators have to peruse the document, even when the raters are selected uniformly at random. We then show that our greedy algorithm outperforms this random selection of raters and other simple greedy algorithms, especially in the later stages of document development. The main metric used to quantify efficiency is human effort, which we assume to be proportional to the number of raters. In other words, the lesser the number of users who have reviewed the document, the more efficient the system. 

\subsection{Network Parameters}
The model that we simulate has many parameters, namely the number $N$ of users, the structure of the underlying collaborative network, the distribution of trust and the threshold values of users denoted by the vector $B$. We now describe the parameters chosen for our experiments.
\begin{itemize}
\item \textbf{Network} We consider directed random or Erd\"{o}s-R{\'e}nyi graphs with $N=10000$ users and edge probability $p$ (i.e. the probability that each of the $N^2$ edges exist). We express $p$ as $\sfrac{D}{N}$, where $D$ is the average number of outgoing links per user.% For convenience, we exclude all self-loops. 
\item \textbf{Trust} Each edge $(i,j) \in E$ has a trust value $t_{ij} \in (0,1)$ associated with it. The trust values are chosen uniformly from the same interval.
\item \textbf{Threshold} All users have a uniform constant threshold value ($b$), except in the final simulation where the threshold values are chosen from a truncated normal distribution with effective mean $0.25$ and variance $0.144$ respectively.
\item \textbf{Raters} In our reference scenario, the raters are chosen randomly with a probability $k$. The expected number of raters is therefore, $kN$. In the final simulation, we compare the performance of an algorithm which selects raters randomly to one which selects the user with the highest incoming trust at each stage to the algorithm based on marginal satisfaction (our Algorithm~\ref{greedyalg}).
\item \textbf{Ratings} We assume all that raters give the article the same rating $r$, that without loss of generality equals $1$.
\item \textbf{Rater Priority} We take the rater weightage factor $\alpha$ to be $0.5$ in all our simulations. As we showed in Equation~(\ref{alphacomparison}), the performance of the system when $\alpha=0.5$ acts as a lower bound to its actual performance for larger values of the parameter.
\item \textbf{Document Content} A key assumption that we make is that the document content does not change throughout this process, i.e. each user only rates the article and does not modify it. This allows us to examine each stage of document development critically.
\end{itemize}

\subsection{Results}
All the experiments were performed in systems which follow the core model proposed in Equation~(\ref{sat2d}) and the Collaborative Editing scheme described in Algorithm~\ref{alg1}. In each case, the simulation was repeated a number times and mean values were plotted in order to obtain smooth curves and average out arbitrary variations.
\subsubsection*{Effect of the number of raters}
\begin{figure}[hbtp]
\begin{center}
\begin{tikzpicture}
\begin{footnotesize}
\begin{axis}[
xlabel=Fraction of raters ($k$),
ylabel=Fraction of users below threshold,
no markers,
domain=0:1,
height=5cm
]
\addplot file {sim/sim1_02.txt};
\addlegendentry{Threshold=0.2}
\addplot file {sim/sim1_03.txt};
\addlegendentry{Threshold=0.3}
\addplot file {sim/sim1_04.txt};
\addlegendentry{Threshold=0.4}
\end{axis}
\end{footnotesize}\end{tikzpicture}
\caption{Fraction of unsatisfied users with increasing number of raters for three different threshold values}
\label{sim1}
\end{center}
\end{figure}
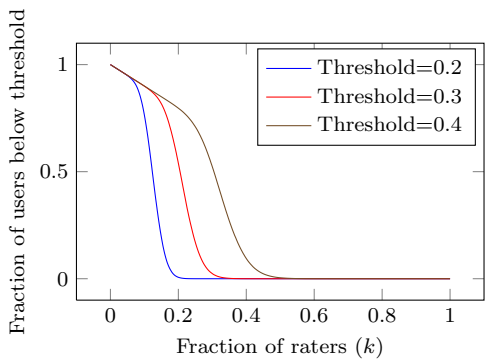
We first investigate the effect of the rater density (or alternatively the total number of raters) on the satisfaction values of non-raters in Figure~\ref{sim1}. We analyze this separately for three different threshold values $b=0.2, b=0.3$ and $b=0.4$ and for edge density $p=\sfrac{50}{10000}$, which results in a dense and well-connected network. As expected, increasing the number of raters decreases the total dissatisfaction in the network. Suppose that we denote by $k_{min}$, the minimum fraction of users required to rate the article so that only $1\%$ of the community remains dissatisfied, which indicates a near linear dependence of $k_{min}$ on the threshold $b$. In other words, increasing the threshold, 	 shifts the curve almost uniformly to the right. The fact that only $0.2$ fraction of the total users have to review an article in order to satisfy the majority of the community points to the efficiency gained due to our model, even when the graph structure is not exploited and raters are selected randomly. 

\subsubsection*{Effect of edge density}
\begin{tiny}
\begin{figure}[hbtp]
\begin{center}
\begin{tikzpicture}
\pgfplotsset{width=4cm,height=4cm}
\matrix {
\begin{axis}[ylabel=Fraction below threshold,
no markers,
title={$k=0.07$},
ymin=0,ymax=1,
]
\addplot file {sim/sim2_07.txt};
\end{axis}
&
% differently large labels are aligned automatically:
\begin{axis}[no markers,title={$k=0.12$},ymin=0,ymax=1]
\addplot file {sim/sim2_12.txt};
\end{axis}
\\
\begin{axis}[ylabel=Fraction below threshold, xlabel=Edge probability ($p$),
no markers,
title={$k=0.15$},
ymin=0,ymax=1
]
\addplot file {sim/sim2_15.txt};
\end{axis}
&
\begin{axis}[xlabel=Edge probability ($p$),
no markers,
title={$k=0.2$},
ymin=0,ymax=1
]
\addplot file {sim/sim2_2.txt};
\end{axis}
\\
};
\end{tikzpicture}
\caption{Impact of Edge density on user satisfaction for different rater densities}
\end{center}
\end{figure}
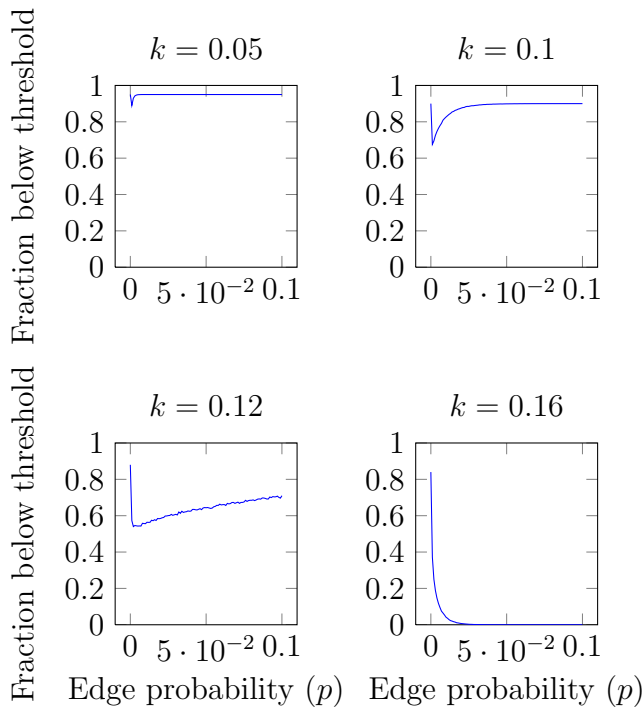
\end{tiny}
We vary the edge probability $p$ from $0$ to $0.1$ and plot the fraction of unsatisfied users for different values of rater densities. The value $p=0.1$ for $N=10000$ indicates that each user has, on an average $1000$ outgoing edges. Beyond this value, the network becomes highly dense and unrealistic. Hence we study the effect of edge probability only up to this value. The satisfaction threshold was uniformly chosen to be $b=0.2$ for all simulations. For lower values of rater density, the number of unsatisfied users in the network sharply decreases as the edge probability increases initially and then settles at a larger value upon further increase. This points to the fact that beyond a point, increasing the number of edges does not improve user satisfaction and has a contrary effect as long as the number of raters is sufficiently low. $k=0.1$ acts as a critical point after which the average satisfaction in the network gradually increases upon increasing the number of raters. Around the $k_{min}$ value for the network ($k=0.16$), increasing the edge probability after a  certain point has little or no effect on user satisfaction as most users remain satisfied. This phenomenon can be explained as follows: when the proportion of raters is low, those raters have a smaller impact in a dense network. Indeed, the nodes that trust a rater also trust several non-raters (and potentially, nodes with no trust path to any rater), which decreases their satisfaction. On the other hand, in a sparse network the raters are trusted only by a few users, but these users are not influenced much by several nodes, and thus are more easily convinced. When the proportion of raters is sufficiently large, the results are in the intuitive direction: a dense network increases the likeliness of trusting a rater, and thus of being satisfied.

\subsubsection*{Performance of greedy algorithms}
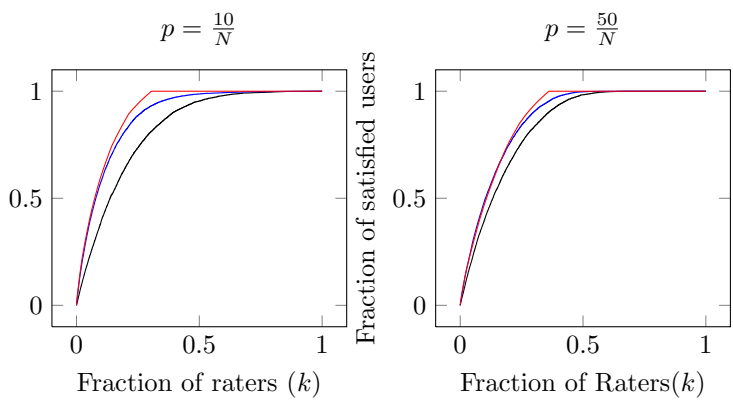
\begin{figure}[hbtp]
\begin{center}
\begin{tikzpicture}
\pgfplotsset{width=5.5cm,height=5cm}
\matrix {
\begin{axis}[xlabel=Fraction of raters ($k$),
no markers,
legend columns=1,
legend entries={Random Raters,Trust Based Greediness, Marginal Satisfaction},
legend to name=named,
title={$p=\frac{10}{N}$}
]
\addplot[color=black] file {sim/sim3_random_p10.txt};
\addplot[color=blue] file {sim/sim3_trust_p10.txt};
\addplot[color=red] file {sim/sim3_greedy_p10.txt};
\end{axis}
&
% differently large labels are aligned automatically:
\begin{axis}[xlabel=Fraction of Raters($k$), ylabel=Fraction of satisfied users,no markers,title={$p=\frac{50}{N}$}]
\addplot[color=black] file {sim/sim3_random_p50.txt};
\addplot[color=blue] file {sim/sim3_trust_p50.txt};
\addplot[color=red] file {sim/sim3_greedy_p50.txt};
\end{axis}
\\
};
\end{tikzpicture}
\begin{center}
\tikzexternaldisable
{\scriptsize \ref{named}}
\tikzexternalenable
\end{center}
\caption{Comparison of three different algorithms to choose raters for two different values of edge probability}
\end{center}
\end{figure}
Finally, we compare the performance of the greedy algorithm based on marginal satisfaction proposed in Section~\ref{maxsatisfaction} to other algorithms for two values of edge probability, $p=\sfrac{10}{10000}$ and $p=\sfrac{50}{10000}$. The trust-based greedy algorithm chooses, at each stage, the user whose sum of incoming trust from other users is maximum, i.e. it selects a user whom other users trust the most (Select the non-rater $i$ such that $\displaystyle\sum\limits_{j \notin N_R}{t_{ji}}$ is maximum over all $i$). The random raters algorithm selects raters uniformly at random from the network. Although initially both the greedy algorithms perform similarly, the trust-based algorithm suffers from slow finish and takes almost twice as many raters to satisfy the entire network as opposed to the marginal satisfaction algorithm.%\rempat{I don't find it very surprising, since at the end it's likely that only isolated users remain unsatisfied: our algorithm will pick them (one by one), but the other algorithms have very little chance of picking them. So maybe it would be more relevant to compare before the 100\% point.}\remshre{No, no. I feel it might also be due to the the small increments in s as number of satisfied users increases. For instance when $>0.7$ of users are satisfied, satisfaction scores of most people are just less than $b$, and thus our algorithm helps in identifying those nodes, which make that little difference}
 This shows that initially the users on whom others have high trust are good choices for raters. However, the underlying dynamics are much more complex as the number of raters increases and the individual effect of each rater depends not only on his trustworthiness but also on the threshold of his neighbors. Our simulations show that our greedy algorithm is more effective for sparse graphs than dense graphs, as increasing the edge probability reduces the difference between the two greedy algorithms. It is also evident that the performance of the random algorithm is quite poor as opposed to those which select raters based on the graph structure and user trust.

\section{Conclusions and Future Work}\label{sec:conclusion}
In this paper, we initiated the study of collaborative document editing systems based on the `web of trust' concept and examined the performance of such systems in large random graphs. We believe that this opens up several avenues for future work given the similarity between collaborative systems and recommendation systems. 

Several of our results were based on the strong assumption that the document content does not change during review. One interesting direction of future research could be to study the dynamics of document editing and the conditions for convergence. This could be done using a model similar to the one proposed by Acemoglu et al~\cite{acemoglu2010spread}, who consider a network where each user holds a belief and users meet according to a Poisson process and exchange their beliefs. A slight modification of this model, with the addition of a common sink representing the document itself and users editing it according to the same Poisson process could be used to study the temporal effects of collaborative editing. Such a model could also be used to update a user's reputation depending on the fraction of the total document that is contributed by him, as proposed in~\cite{phanle2009evaluating}.

Closely linked to the above paradigm is the game-theoretic analysis of collaborative editing. The whole process can be viewed as a repeated game between several collaborators whose strategies correspond to content to add or modify in the document. Alternatively, one could also design incentives or mechanisms to induce editors to reveal their true evaluation of a document. For instance, users could be required to submit their evaluations (say $r_i$) for the document under consideration and their utility could be a function of ($|r_{avg} - r_i|$), where $r_{avg}$ is the average of the evaluations of all users. Such games would however require efficient modeling of user preferences and document content and a consideration of which of the equilibria are actually reachable in practice. 

We quantified trust between users as a value in the range $[0,1]$, thereby avoiding negative trust or distrust between users. Propagation of negative trust not only adds to the complexity of the model but also leads to certain inconsistencies as notions that are valid for trust may not be so realistic in the case of distrust. If user $i$ distrusts user $j$, who in turn distrusts $k$, a simple propagation model would indeed result in user $i$ trusting $k$. However, the notion `the enemy of my enemy is my friend' may not always be valid in real life. Guha et al~\cite{guha2004propagation} try tackling some of these issues via two separate matrices capturing trust and distrust respectively. One possible direction of future research could involve checking whether our results are valid when negative trust is introduced. We have also deliberately avoided the effect of malicious users in our system and focused more on the dynamics of recognizing vandals. It would be interesting to study convergence in a system where a small fraction of users have malicious intent.

We showed the efficiency of our system in Erd\"{o}s-R{\'e}nyi graphs with conservative values of the edge probability parameter $p$. Whether this efficiency is improved or worsened in more realistic models like Small World~\cite{watts1998collective} or the scientific co-authorship graph~\cite{newman2001structure} warrants further experiments. %Finally, one interesting application for our model could be a network of academicians reviewing a submitted paper for publication in a journal or conference. Our results ensure that the ratings of a small fraction of the total reviewer set suffices to judge the quality of the document.

\section*{Acknowledgement}

This work has been partially funded by the French research agency, through the FLUOR project (\url{http://fluor.no-ip.fr/index.php?lang=en}).

%\printbibliography
\bibliographystyle{ieeetr}
\bibliography{bibliography}

\begin{thebibliography}{10}

\bibitem{vintagesew}
Vintage sewing patterns, July 2011.

\bibitem{vandalismwiki}
Wikipedia: Vandalism statistics, July 2011.

\bibitem{starwars}
Wookiepedia, the {S}tar {W}ars wiki, July 2011.

\bibitem{acemoglu2010spread}
D.~Acemoglu, A.~Ozdaglar, and A.~Parandeh~Gheibi.
\newblock Spread of (mis)information in social networks.
\newblock {\em Games and Economic Behavior}, 70(2):194--227, Nov 2010.

\bibitem{adler2008assigning}
B.T. Adler, K.~Chatterjee, L.~De~Alfaro, M.~Faella, I.~Pye, and V.~Raman.
\newblock Assigning trust to wikipedia content.
\newblock In {\em Proc. of 4th ACM International Symposium on Wikis}, pages
  1--12, Porto, Portugal, Sept 2008.

\bibitem{andersen2008trust}
R.~Andersen, C.~Borgs, J.~Chayes, U.~Feige, A.~Flaxman, A.~Kalai, V.~Mirrokni,
  and M.~Tennenholtz.
\newblock Trust-based recommendation systems: an axiomatic approach.
\newblock In {\em Proc. of 17th ACM international conference on World Wide
  Web}, pages 199--208, 2008.

\bibitem{asavathiratham2001influence}
C.~Asavathiratham, S.~Roy, B.~Lesieutre, and G.~Verghese.
\newblock The influence model.
\newblock {\em IEEE Control Systems Magazine}, 21(6):52--64, 2001.

\bibitem{barabási2000scale}
A.L. Barab{\'a}si, R.~Albert, and H.~Jeong.
\newblock Scale-free characteristics of random networks: the topology of the
  world-wide web.
\newblock {\em Physica A: Statistical Mechanics and its Applications},
  281(1-4):69--77, 2000.

\bibitem{barrett1994templates}
R.~Barrett.
\newblock {\em Templates for the solution of linear systems: building blocks
  for iterative methods}.
\newblock Society for Industrial Mathematics, 1994.

\bibitem{blaze1999role}
M.~Blaze, J.~Feigenbaum, J.~Ioannidis, and A.D. Keromytis.
\newblock {\em Secure Internet Programming}, chapter The role of trust
  management in distributed systems security, pages 185--210.
\newblock Springer, 1999.

\bibitem{blaze1996managing}
M.~Blaze, J.~Feigenbaum, and J.~Lacy.
\newblock Managing trust in medical information systems.
\newblock Technical Report 96.14. 1, AT\&T, 1996.

\bibitem{bollobas2001random}
B.~Bollob{\'a}s.
\newblock {\em Random graphs}.
\newblock Cambridge University Press, 2001.

\bibitem{chang2009spreading}
C.L. Chang and Y.D. Lyuu.
\newblock Spreading messages.
\newblock {\em Theoretical Computer Science}, 410(27-29):2714--2724, 2009.

\bibitem{craig2007changing}
E.M. Craig.
\newblock Changing paradigms: managed learning environments and web 2.0.
\newblock {\em Campus-Wide Information Systems}, 24(3):152--161, 2007.

\bibitem{domingos2001mining}
P.~Domingos and M.~Richardson.
\newblock Mining the network value of customers.
\newblock In {\em Proc. of 7th ACM International Conference on Knowledge
  Discovery and Data Mining}, San Francisco, CA, USA, Aug 2001.

\bibitem{durrett2010some}
R.~Durrett.
\newblock Some features of the spread of epidemics and information on a random
  graph.
\newblock {\em Proc. of the National Academy of Sciences}, 107(10):4491, 2010.

\bibitem{erdos1960evolution}
P.~Erd{\H{o}}s and A.~R{\'e}nyi.
\newblock On the evolution of random graphs.
\newblock Publication of the Mathematical Institute of the Hungarian Academy of
  Sciences, 1960.

\bibitem{fang1997worst}
X.G. Fang and G.~Havas.
\newblock On the worst-case complexity of integer gaussian elimination.
\newblock In {\em Proc. of ACM International Symposium on Symbolic and
  Algebraic Computation}, pages 28--31, Maui, Hawaii, USA, 1997.

\bibitem{ganesh2005effect}
A.~Ganesh, L.~Massouli{\'e}, and D.~Towsley.
\newblock The effect of network topology on the spread of epidemics.
\newblock In {\em Proc. of IEEE INFOCOM}, pages 1455--1466, 2005.

\bibitem{golbeck2005computing}
J.A. Golbeck.
\newblock {\em Computing and applying trust in web-based social networks}.
\newblock PhD thesis, University of Maryland, College Park, MD, USA, 2005.

\bibitem{gruhl2004information}
D.~Gruhl, R.~Guha, D.~Liben-Nowell, and A.~Tomkins.
\newblock Information diffusion through blogspace.
\newblock In {\em Proc. of 13th ACM International Conference on World Wide
  Web}, pages 491--501, 2004.

\bibitem{guha2004propagation}
R.~Guha, R.~Kumar, P.~Raghavan, and A.~Tomkins.
\newblock Propagation of trust and distrust.
\newblock In {\em Proc. of 13th ACM International Conference on World Wide
  Web}, pages 403--412, 2004.

\bibitem{hoegg2006overview}
R.~Hoegg, R.~Martignoni, M.~Meckel, and K.~Stanoevska-Slabeva.
\newblock Overview of business models for web 2.0 communities.
\newblock {\em Communication}, 2006:1--17, 2006.

\bibitem{josang2007survey}
A.~Josang, R.~Ismail, and C.~Boyd.
\newblock A survey of trust and reputation systems for online service
  provision.
\newblock {\em Decision Support Systems}, 43(2):618--644, 2007.

\bibitem{kempe2003maximizing}
D.~Kempe, J.~Kleinberg, and {\'E}.~Tardos.
\newblock Maximizing the spread of influence through a social network.
\newblock In {\em Proc. of 9th ACM SIGKDD International Conference on Knowledge
  Discovery and Data mining}, pages 137--146, 2003.

\bibitem{lappas2010finding}
T.~Lappas, E.~Terzi, D.~Gunopulos, and H.~Mannila.
\newblock Finding effectors in social networks.
\newblock In {\em Proc. of 16th ACM International Conference on Knowledge
  Discovery and Data Mining}, pages 1059--1068, Washington, DC, USA, 2010.

\bibitem{lohr2009netflix}
S.~Lohr.
\newblock Netflix awards \$1 million prize and starts a new contest.
\newblock
  \url{http://bits.blogs.nytimes.com/2009/09/21/netflix-awards-1-million-prize-and-starts-a-new-contest/},
  2009.

\bibitem{massa2005controversial}
P.~Massa and P.~Avesani.
\newblock Controversial users demand local trust metrics: An experimental study
  on epinions.com community.
\newblock In {\em Proc. of 20th National Conference on Artificial
  Intelligence}, volume~20, page 121, Pittsburgh, PA, USA, 2005.

\bibitem{moler2004numerical}
C.B. Moler.
\newblock {\em Numerical computing with MATLAB}.
\newblock Society for Industrial Mathematics, 2004.

\bibitem{newman2001structure}
M.~E.~J. Newman.
\newblock The structure of scientific collaboration networks.
\newblock {\em Proc. of the National Academy of Sciences}, 98(2):404, 2001.

\bibitem{phanle2009evaluating}
C.~T. Phan~Le, F.~Cuppens, N.~Cuppens, and P.~Maill\'e.
\newblock Evaluating the trustworthiness of contributors in a collaborative
  environment.
\newblock In {\em Proc. of CollaborateCom}, Orlando, FL, USA, Nov 2008.

\bibitem{resnick2000reputation}
P.~Resnick, K.~Kuwabara, R.~Zeckhauser, and E.~Friedman.
\newblock Reputation systems.
\newblock {\em Communications of the ACM}, 43(12):45--48, 2000.

\bibitem{resnick2002trust}
P.~Resnick and R.~Zeckhauser.
\newblock {\em The Economics of the Internet and E-Commerce (Advances in
  Applied Microeconomics)}, volume~11, chapter Trust among strangers in
  Internet transactions: Empirical analysis of {eBay}'s reputation system,
  pages 127--157.
\newblock Elsevier, 2002.

\bibitem{richardson2003trust}
M.~Richardson, R.~Agrawal, and P.~Domingos.
\newblock Trust management for the semantic web.
\newblock {\em Proc. of 2nd Semantic Web Conference}, pages 351--368, 2003.

\bibitem{rozenfeld2009consistent}
O.~Rozenfeld and M.~Tennenholtz.
\newblock Consistent continuous trust-based recommendation systems.
\newblock In {\em Proc. of 5th International Workshop on Internet and Network
  Economics}, pages 113--124. Springer, 2009.

\bibitem{snyder2007sa}
J.~Snyder.
\newblock It's a wiki-world: Utilizing {W}ikipedia as an academic reference.
\newblock In {\em Proc. of Mountain Plains Management Conference}, 2007.

\bibitem{sun2008performance}
T.~Sun and M.K. Denko.
\newblock Performance evaluation of trust management in pervasive computing.
\newblock In {\em Proc. of 22nd IEEE International Conference on Advanced
  Information Networking and Applications}, pages 386--394, 2008.

\bibitem{sztompka1999trust}
P.~Sztompka.
\newblock {\em Trust: A sociological theory}.
\newblock Cambridge University Press, 1999.

\bibitem{vojnovic2008race}
M.~Vojnovic and A.J. Ganesh.
\newblock On the race of worms, alerts, and patches.
\newblock {\em IEEE/ACM Transactions on Networking (TON)}, 16(5):1066--1079,
  2008.

\bibitem{walter2008model}
F.E. Walter, S.~Battiston, and F.~Schweitzer.
\newblock A model of a trust-based recommendation system on a social network.
\newblock {\em Autonomous Agents and Multi-Agent Systems}, 16(1):57--74, 2008.

\bibitem{watts1998collective}
D.J. Watts and S.H. Strogatz.
\newblock Collective dynamics of small-world networks.
\newblock {\em Nature}, 393(6684):440--442, 1998.

\bibitem{wikia}
Wikia.
\newblock Wiki communities for everyone!, July 2011.

\bibitem{wilhelm1998problem}
U.G. Wilhelm, S.~Staamann, and L.~Butty{\'a}n.
\newblock On the problem of trust in mobile agent systems.
\newblock In {\em Proc. of Internet Society's Symposium on Network and
  Distributed System Security}, pages 114--124, San Diego, CA, USA, 1998.

\bibitem{wortman2008viral}
J.~Wortman.
\newblock Viral marketing and the diffusion of trends on social networks.
\newblock Technical Report MS-CIS-08-19, Department of Computer \& Information
  Science Technical Reports (CIS), University of Pennsylvania, PA, USA, May
  2008.

\bibitem{xiong2004peertrust}
L.~Xiong and L.~Liu.
\newblock Peertrust: Supporting reputation-based trust for peer-to-peer
  electronic communities.
\newblock {\em IEEE Transactions on Knowledge and Data Engineering},
  16(7):843--857, 2004.

\bibitem{maille2011trust}
X.~Zheng, P.~Maill\'e, C.~Phan~Le, C.~T., and S.~Morucci.
\newblock Improving the efficiency of collaborative work with trust management.
\newblock In {\em Proc. of 1st IFIP/IEEE Workshop on Managing Federations and
  Cooperative Management}, Dublin, Ireland, May 2011.

\end{thebibliography}

\begin{appendices}
\section{Trust Updation Model}
In the previous sections, we defined the idea of trust for collaborative-working contexts and proposed a satisfaction estimation model based on this. The model and associated algorithms however, assume the existence of trust between various pairs of users and do not dwell upon how to determine the trust between users. We now propose a generic trust updation mechanism that calculates the new trust value, given the old trust between a pair of users, and their current ratings. Our model is motivated by the following observations:
\begin{enumerate}
\item Content in most collaborative systems cannot be purely objective as it is provided by human editors. Therefore, the trust between editors is also subjective and depends on their individual opinions. A user may greatly mistrust another user whose opinions do not concur with his own. We conclude that the users with similar outlook and content should have high mutual trust and vice-versa. This is in congruence with the definition of trust that we used in Section~\ref{trustdefinition} that trust is the amount of faith a user has in the actions of others.

\item We mentioned earlier that the social network structure in the system may be explicitly provided or derived implicitly. Indeed, the underlying network of most collaborative systems cannot be directly expressed in terms of friends, contacts and followers and has to be calculated using known information. We aim to preserve this property of web-based collaborative systems and estimate trust between users by using their previous interactions with respect to a common article of reference. When no such standard exists, we assume that the trust between the users is zero.

\item Updating trust based on a user's overall contribution to the published document while objective, is more suited for reputation systems. For instance, in~\cite{phanle2009evaluating}, the overall trust value of a user (say $i$) is calculated as the fraction of the total number of words in the article, written by that particular user. Such systems depend somewhat on semantics, and hence can be \textit{cheated} to a certain extent. In the above model, a user may deliberately rewrite simple sentences in a complex manner in order to increase his trust/reputation.
\end{enumerate}

We use the following key idea for updating trust: a user $i$'s trust $t_{ij}$ in another user $j$ should increase if the rating of user $i$ is similar to that of user $j$ at a particular stage of document development. The farther the rating of user $i$ is from $j$'s rating, the more $j$ will reduce its trust in $i$. Stating this mathematically, the trust value should be inversely dependent on $|r_i - r_j|$. Hence, the general updation formula should be as follows:
\begin{equation}
t_{ij}' = \gamma t_{ij} + (1- \gamma) f(|r_i - r_j|),
\label{trustgeneric}
\end{equation}
where $t_{ji}'$ is the new trust value of user $i$ in user $j$ and $t_{ij}$ the old value, $\gamma$ is a parameter which determines the rate of change of trust and $f$ is a monotonically decreasing function. Using this model, the updation mechanism could be described as follows:
\begin{quote}
Everytime a user ($i$) reads, and rates the article, update the trust values ($t_{ij}$ and $t_{ji}$) for all $j \in N_R$ using Equation~(\ref{trustgeneric}).
\end{quote}
This mechanism results in \textit{fairness} as when a user with few or no edges rates the article, trust values between this user and other existing raters are automatically calculated without any direct interaction. Also, by \textbf{Independence} (between rater and rater) property, this trust updation model does not violate any of the other properties elucidated in Section~\ref{model}. Now before choosing a specific function for updating trust, we list the requirements that any model must follow:
\begin{enumerate}
\item The function must be continuous and monotonically decreasing. 
 
\item $0 \leq |r_i - r_j| \leq 1$. Thus our function $f$ must satisfy the boundary conditions, $f(0) = 1$ and $f(1)=0$ as $t_{ij}$ also lies in the same range.
\end{enumerate}

\subsection*{Formula}
The simplest and most obvious function which satisfies our constraints is the linear function $f(x) = 1-x$. Here we do not bother ourselves with the function behavior outside the domain $[0,1]$.  The following graph displays the function in the required domain. The overall updation formula becomes
\[
t_{ij}' = \gamma t_{ij} + (1- \gamma) (1 - x),
\]
where $x=|r_i - r_j|$. 
%\begin{wrapfigure}{r}{0.5\textwidth}
%      \centering
%\begin{tikzpicture}
%\begin{axis}[
%xlabel=$x$,
%ylabel=$1-x$,
%grid=minor,
%no markers,
%width=6cm,
%height=6cm,
%domain=0:1
%]
%% invoke external gnuplot as
%% calculator:
%%\addplot gnuplot[id=lin]{1-x};
%\addplot {1-x};
%\end{axis}
%\end{tikzpicture}
%\caption{Linear Function $f(x)=1-x$}
%\label{lineartrust}
%\end{wrapfigure}
The problem with this function is that slope of the linear function is too gradual, and it gives relatively high values even when there is a sufficient difference in the ratings. For example, when $x=0.5$, $f(x) =0.5$ and if the previous trust between the users was less than this value, then this increases the trust between them. This is however, not very practical as such a large difference in rating indicates a significant difference of opinion. We need a function that gives a high trust value when $x$ is close to $0$ but decreases rapidly after that. It is evident from this that a convex function would be appropriate for our needs.

A general convex function that (asymptotically) fits our needs is 
\begin{equation}
f(x) = \frac{1}{1 + px}
\end{equation} 
\begin{figure}[hbtp]
\begin{center}
\begin{tikzpicture}
\begin{axis}[
xlabel=$x$,
ylabel=$\frac{1}{1+px}$,
grid=major,
no markers,
domain=0:1
]
% invoke external gnuplot as
% calculator:
%\addplot gnuplot[id=p3]{1/(1+3*x)};
%\addplot gnuplot[id=p7]{1/(1+7*x)};
%\addplot gnuplot[id=p16]{1/(1+16*x)};
\addplot {1/(1+3*x)};
\addplot {1/(1+7*x)};
\addplot {1/(1+16*x)};
\legend{$p=3$,$p=7$,$p=16$}
\end{axis}
\end{tikzpicture}
\caption{Function $\frac{1}{1+px}$ plotted for three different values of $p$}
\label{convextrust}
\end{center}
\end{figure}
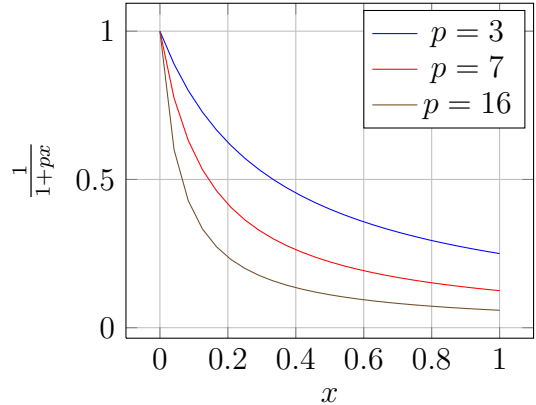
The larger the value of $p$, the more convex the function is. Moreover, as $ p \rightarrow \infty$, $f(1) \rightarrow 0$. The function has been plotted for different values of $p$ in Figure~\ref{convextrust}. The value of $p$ can be chosen based on parameters like the total size of the community, the definition of trust as is appropriate for the task at hand and the value of $\gamma$ that is to be considered.

The advantage of using such a model is that one can build the whole trust matrix based solely on the revision history of the article. Such a model, when implemented centrally also provides some incentive against vandalism and manipulation. Firstly, users who deliberately provide the article an incorrect rating to further their own motives incur the wrath of low trust from other users. Moreover, any algorithm that selects raters based on trust such as the one described in Section~\ref{maxsatisfaction} would never select such a user due to his poor impact on other collaborators. Secondly, it is possible for a user to manipulate the system by providing a rating similar to that of other users in order to increase his \textit{trustworthiness}. This is avoided by keeping ratings private thereby encouraging good behavior. 

\end{appendices}

\end{document}